\documentclass[letterpaper]{article} % DO NOT CHANGE THIS
\usepackage[preprint]{aaai2027}
\usepackage[hyphens]{url}  % DO NOT CHANGE THIS
\usepackage{graphicx} % DO NOT CHANGE THIS
\usepackage{natbib}  % DO NOT CHANGE THIS AND DO NOT ADD ANY OPTIONS TO IT
\usepackage{caption} % DO NOT CHANGE THIS AND DO NOT ADD ANY OPTIONS TO IT
\usepackage{amsmath}
\usepackage{amssymb}
\usepackage{algorithm}
\usepackage{algorithmic}
\usepackage{comment}
\usepackage{cleveref}

\newtheorem{theorem}{Theorem}[section]
\newtheorem{definition}[theorem]{Definition}
\newtheorem{lemma}[theorem]{Lemma}
\newtheorem{corollary}[theorem]{Corollary}
\newtheorem{proposition}[theorem]{Proposition}
\newtheorem{assumption}[theorem]{Assumption}
\newtheorem{remark}[theorem]{Remark}
\newenvironment{proof}[1][Proof]{\par\noindent\textit{#1.}\ }{\hfill$\square$\par}

\newcommand{\cC}{\mathcal{C}}
\newcommand{\cF}{\mathcal{F}}
\newcommand{\cM}{\mathcal{M}}
\newcommand{\cO}{\mathcal{O}}
\newcommand{\cR}{\mathcal{R}}

\newcommand{\cX}{\mathcal{X}}
\newcommand{\bbZ}{\mathbb{Z}}
\newcommand{\bpi}{\boldsymbol{\pi}}
\newcommand{\bs}{\boldsymbol{s}}
\newcommand{\bv}{\boldsymbol{v}}
\newcommand{\score}{\mathsf{score}}
\newcommand{\scoreplu}{\mathsf{score}_{\mathrm{plu}}}
\newcommand{\tscoreplu}{\widetilde{\scoreplu}}
\newcommand{\activecandidates}{\mathcal{A}}
\newcommand{\mov}{\operatorname{MoV}}
\newcommand{\move}{\operatorname{Move}}

\newcommand{\bribery}{\ensuremath{\text{\scshape Bribery}}}
\newcommand{\alg}{\mathcal{M}}
\newcommand{\eps}{\varepsilon}
\newcommand{\OPT}{\operatorname{OPT}}
\newcommand{\ox}{\bar{x}}
\newcommand{\oell}{\bar{\ell}}
\newcommand{\coloneq}{\mathrel{\mathop:}=}

\newcommand{\Paren}[1]{\left(#1\right)}

\newcommand{\tred}[1] {\textcolor{red}{#1}}
\newcommand{\tblue}[1] {\textcolor{blue}{#1}}
\newcommand{\torange}[1] {\textcolor{orange}{#1}}

\newcommand{\scmaximin}{\mathsf{score}_{\mathrm{maximin}}}
\newcommand{\tscmaximin}{\widetilde{\mathsf{score}}_{\mathrm{maximin}}}
\newcommand{\argmin}{\mathrm{argmin}}

\iftrue
\newcommand{\pasin}[1]{\tblue{[Pasin: #1]}}
\newcommand{\vorapong}[1]{\torange{[Vorapong: #1]}}
\newcommand{\phanu}[1]{\tred{[Phanu: #1]}}
\else
\newcommand{\pasin}[1]{}
\newcommand{\vorapong}[1]{}
\newcommand{\phanu}[1]{}
\fi

\title{Decisive Margins in Differentially Private Voting}
\author {
    Quentin Hillebrand\textsuperscript{\rm 1},
    Pasin Manurangsi\textsuperscript{\rm 2},
    Vorapong Suppakitpaisarn\textsuperscript{\rm 3},
    Phanu Vajanopath\textsuperscript{\rm 4}
}
\affiliations {
    \textsuperscript{\rm 1}University of Copenhagen, Denmark\\
    \textsuperscript{\rm 2}Google Research, Thailand\\
    \textsuperscript{\rm 3}The University of Tokyo, Japan\\
    \textsuperscript{\rm 4}University of Wrocław, Poland
}

\begin{document}

\maketitle

\begin{abstract}
Differential privacy protects individual voting records by injecting randomness
into the published outcome, but this noise can lead to erroneous results when an
election is close.  We study how precise central differential privacy and local
differential privacy can be for common voting rules, including Plurality,
Condorcet, Maximin, Plurality with Runoff, and Single Transferable Vote (STV).  Our measure
of precision is the margin of victory needed for a private mechanism to return
the same winner as the non-private rule with high probability.  We give private
algorithms for publishing the winner and prove upper bounds on the required
margin for these algorithms.  We also prove lower bounds showing that nontrivial
margins are necessary; many of these bounds match the corresponding upper bounds
up to logarithmic factors.  For STV, an information-theoretic upper bound matches
the lower bound, but we prove that this guarantee cannot be achieved in
polynomial time unless NP $\subseteq$ BPP.  This gives a rare example of a computationally tractable task that becomes intractable
when one simultaneously requires differential privacy and utility.
\end{abstract}

\section{Introduction}
\label{sec:introduction}

Differential privacy has become a standard formal framework for protecting
individual-level data in statistical analysis.  It requires that the published
result of an algorithm reveal only limited information about any one individual's
record, even to an adversary with arbitrary side information.  This guarantee is
especially important when the data being aggregated is sensitive, because it
allows useful population-level information to be released while limiting the
privacy risk to each participant~\cite{DworkMNS06,DworkRoth14}.

Local differential privacy (LDP) strengthens this protection by removing the
need for a trusted central curator.  In the local model, which goes back to randomized
response and was later formalized in the differential-privacy literature, each
individual randomizes her own record before sending any message to the
aggregator~\cite{Warner65,kasiviswanathan2011can,ErlingssonPK14}.  This makes
LDP particularly attractive in settings where the data collector should not
see raw user data, at the cost of injecting more noise than in the central
model.

Voting is a natural setting where these privacy guarantees are important.  A
single ranking can reveal political views, strategic preferences, or
institutional priorities, so publishing election information without privacy can
expose sensitive individual preferences.  Prior work has therefore studied
privacy for voting from several perspectives, including LDP vote
aggregation, differentially private rank aggregation, the privacy of commonly
used voting rules, and differentially private voting
rules~\cite{WangDWDLLNH19,AlabiGKM22,HillebrandMSV26,LiuLXZ20,LiLXCW23}.  For voting,
however, privacy noise has a particularly sharp effect: a small amount of noise
can change the declared winner when the election is close.  Thus, even if the
non-private voting rule is easy to evaluate, a private mechanism may return an
erroneous winner unless the election is sufficiently decisive.

\paragraph{Our contributions.} This paper asks how decisive an election must be before private winner
publication is reliable.  We study this question for central differential
privacy and local differential privacy, focusing on common voting rules such as
Plurality, Condorcet, Maximin, Plurality with Runoff, and Single Transferable
vote (STV).
Our measure of precision is the \emph{margin of victory} (MoV)~\cite{Xia12,BHATTACHARYYA2021103476}, the minimum
number of voters whose preferences must be changed to alter the result.  

In other words, we study the following question:
\vspace{0.1cm}
\begin{quote}
{\em What is the smallest $\mov$ required for an election, so that the correct outcome can be differentially privately published with high probability?}
\end{quote}
\vspace{0.2cm}

\begin{table*}[ht]
\centering
\footnotesize
\setlength{\tabcolsep}{2pt}
\renewcommand{\arraystretch}{1.08}
\begin{tabular}{@{}p{0.27\textwidth}p{0.11\textwidth}p{0.11\textwidth}p{0.11\textwidth}p{0.22\textwidth}p{0.11\textwidth}@{}}
\hline
\textbf{Setting} & \textbf{Plurality} & \textbf{Maximin} & \textbf{Runoff} & \textbf{STV} & \textbf{Condorcet} \\
\hline
Central DP upper bounds & $O(\log m)$ & $O(\log m)$ & $O(\log m)$ & $O(\log m)$\textsuperscript{\dag} & $O(\log m)$ \\
Central DP lower bounds & $\Omega(\log m)$ & $\Omega(\log m)$ & $\Omega(\log m)$ & $\Omega(\log m)$ & $\Omega(\log m)$ \\
Non-interactive LDP upper bound & $\tilde{O}(\sqrt{n})$ & $\tilde{O}(m\sqrt{n})$ & $\tilde{O}(m\sqrt{n})$ & $\tilde{O}(2^{m/2}\sqrt{n})$ & $\tilde{O}(m\sqrt{n})$ \\
Non-interactive LDP lower bound & $\tilde{\Omega}(\sqrt{n})$ & $\tilde{\Omega}(m\sqrt{n})$ & $\tilde{\Omega}(m\sqrt{n})$ & $\tilde{\Omega}(\min\{n/m,2^{\Omega(\sqrt{m})}  \sqrt{n}\})$ & $\tilde{\Omega}(m\sqrt{n})$ \\
Interactive LDP upper bound & $\tilde{O}(\sqrt{n})$ & $\tilde{O}(m\sqrt{n})$ & $\tilde{O}(\sqrt{n})$ & $\tilde{O}(m^2\sqrt{n})$ & $\tilde{O}(\sqrt{mn})$ \\
Interactive LDP lower bound & $\tilde{\Omega}(\sqrt{n})$ & $\tilde{\Omega}(\sqrt{n})$ & $\tilde{\Omega}(\sqrt{n})$ & $\tilde{\Omega}(\sqrt{n})$ & $\tilde{\Omega}(\sqrt{n})$ \\
\hline
\end{tabular}
\caption{Summary of margin-of-victory thresholds for private winner publication. Here $n$ is the number of voters and $m$ is the number of candidates.
Dependencies on the privacy budgets and failure probabilities are omitted for brevity.
%Privacy budgets and upper-bound failure probabilities are treated as constants. 
The STV central-DP upper bound marked by \textsuperscript{\dag} is information-theoretic; under NP $\nsubseteq$ BPP, it cannot be achieved in polynomial time.}
\label{tab:results-summary}
\end{table*}

Table~\ref{tab:results-summary} summarizes the main results.
Our first contribution is a set of private algorithms for publishing election
winners together with upper bounds on the margin they require.  In the central
model, the exponential mechanism gives a general logarithmic-MoV guarantee.  For
rules where the relevant margin quantities can be computed efficiently, this
also yields efficient algorithms; for Maximin, we use a polynomial-time
approximation for the exact MoV.  In the local model, where each voter privatizes
their own message before aggregation, we give non-interactive and interactive
protocols based on noisy plurality scores, pairwise statistics, and STV
elimination statistics.

Our second contribution is a set of lower bounds showing that nontrivial margins
are necessary.  In the central model, packing arguments give matching
logarithmic lower bounds for broad classes of rules.  In the local model,
testing and marginal-release reductions imply that the usefulness threshold must
grow with $\sqrt{n}$, up to logarithmic factors, for several rules.  Many of
these lower bounds match the corresponding upper bounds up to logarithmic or
rule-dependent factors.

For LDP, a highlight of our results is ``the cost of non-interactivity'', a provable privacy--interactivity
trade-off for Condorcet, Plurality with Runoff, and STV.  For these rules,
non-interactive algorithms require a larger margin of victory than interactive
algorithms.  The separation is especially striking for STV: interaction
improves the dependence on the number of candidates superpolynomially.  
In the
nonsaturated regime, the non-interactive lower bound is
$2^{\Omega(\sqrt{m})}\sqrt{n}$, whereas the interactive upper bound is
$\widetilde{O}(m^2\sqrt{n})$.  This is a significant departure
from the non-private setting, where a voter can report her entire ranking in a
single round.

For central DP, the most remarkable case among our results is once again STV.  Information-theoretically, the generic central
DP upper bound matches the lower bound.  Computationally, however,
this precision is out of reach: assuming NP $\nsubseteq$ BPP, no polynomial-time
central-DP algorithm can achieve the same high-precision guarantee for STV (\Cref{thm:stv-no-efficient-dp}). This gives a rare example of a computationally tractable task that becomes intractable when one simultaneously requires differential
privacy and high-precision recovery of the non-private result. Previously, such results are only known in specific settings such as synthetic data generation and high-probability learning tasks~\cite{dwork2009complexity,UllmanV20,GeorgievH22,GhaziGKKKM26}. We hope that our proof--which relates NP-hardness of approximation to the intractability of DP algorithms--can contribute to further study on this broader line of research.

%This result also suggests a broader research line on computational
%intractability induced by differential privacy.  Related phenomena are known for
%private data release and synthetic data generation~\cite{dwork2009complexity,UllmanV20},
%but the source of hardness here is different: publishing the STV winner becomes
%intractable because the sensitivity-like margin quantity governing
%high-precision private recovery is hard to approximate.  We suspect that similar
%arguments can turn inapproximability of rule-specific stability or sensitivity
%measures into new DP intractability results.
\vspace{0.1cm}
\noindent
\textbf{Organization.} We next provide preliminaries. Then, we present our results for central DP in \Cref{sec:central-dp}. Our LDP results are separated into two sections, \Cref{sec:noninteractive-ldp} for non-interactive LDP and \Cref{sec:interactive-ldp} for interactive LDP. Due to space constraints, we provide proof sketches in the main body and full proofs are deferred to the appendix.

\section{Preliminaries}

\subsection{Voting}

%Let $\cC$ be the candidate set, $m=|\cC|$, and write
Let $\cC=\{c_1,\ldots,c_m\}$ be the candidate set, and $\Pi_{\cC}$ be the set of all %strict 
rankings of the candidates.  A \emph{preference profile} is a tuple
$\bpi=(\pi_1,\ldots,\pi_n)\in\Pi_{\cC}^n$, where $\pi_i$ is voter $i$'s ranking.
A %deterministic 
\emph{voting rule} is a function
$r:\Pi_{\cC}^*\to \cC\cup\{\bot\}$ that returns either a winning candidate or
$\bot$ when no winner exists.

Two profiles $\bpi=(\pi_1,\ldots,\pi_n)$ and
$\bpi'=(\pi'_1,\ldots,\pi'_n)$ are \emph{$k$-neighbors} if they differ in at most
$k$ voters' rankings: there exists $S\subseteq[n]$ with $|S|\le k$ such that
$\pi_i=\pi'_i$ for every $i\notin S$.  We call $1$-neighbors simply
\emph{neighboring profiles}.

\subsubsection{Voting rules.}

\label{subsubsec:votingRules}

We consider the following voting rules. %throughout the paper.
\begin{itemize}
\item \textbf{Plurality: }
For a profile $\bpi=(\pi_1,\ldots,\pi_n)$ and each candidate $c\in\cC$, let
$\mathsf{score}_{\mathrm{plu}}(c)$
be the number of voters who rank $c$ first.  Plurality outputs a candidate
$c^\star \in \arg\max_{c\in\cC}\mathsf{score}_{\mathrm{plu}}(c)$.

\item \textbf{Scoring-based rules: } A scoring rule is parameterized by a vector
$\bs=(s_1,\ldots,s_m)$ with $s_1\ge\cdots\ge s_m\ge0$.  A candidate ranked in
position $j$ by a voter receives score $s_j$, and the candidate with the largest
total score wins.  Formally,
$\score_{\bs}(c; \bpi)
    := \sum_{i=1}^n \sum_{j=1}^m s_j \cdot \mathbf{1}[\pi_i(j)=c]$.

Plurality is the scoring rule with $\bs=(1,0,\ldots,0)$.  Other examples include
Borda, with $\bs=(m-1,m-2,\ldots,0)$, and Veto, with
$\bs=(1,1,\ldots,1,0)$.

\item \textbf{Maximin: }
For candidates $a,b\in\cC$, let $N_{\bpi}(a\succ b)$ be the number of voters who
rank $a$ above $b$, and define the pairwise margin
$D_{\bpi}(a,b)=N_{\bpi}(a\succ b)-N_{\bpi}(b\succ a)$.  The maximin score of
$a$ is
    $\mathsf{score}_{\mathrm{maximin}}(a)=\min_{b\in\cC\setminus\{a\}}D_{\bpi}(a,b)$.
Maximin outputs a candidate with maximum maximin score.

\item \textbf{Condorcet: }
The Condorcet rule outputs a candidate $c$ if $D_{\bpi}(c,a)>0$ for every
$a\in\cC\setminus\{c\}$; if no such candidate exists, it outputs $\bot$.

\item \textbf{Plurality with Runoff: }
Plurality with Runoff proceeds in two rounds.
In the first round, we compute each candidate's plurality score
$\mathsf{score}_{\mathrm{plu}}(c)$
and select the two candidates with the highest plurality scores; %(breaking ties by a fixed tie-breaking rule); 
denote them by $c_1,c_2$.
In the second round, we conduct a head-to-head majority vote between $c_1$ and $c_2$:
we output $c_1$ if $N_{\bpi}(c_1\succ c_2) > N_{\bpi}(c_2\succ c_1)$, and otherwise output $c_2$.

\item \textbf{Single Transferable Vote (STV): } To define STV, it is convenient to introduce a notation: for every set $S \subseteq \cC$, let $\pi|_S \in \Pi_{S}$ denote the preference $\pi$ when restricted to $S$.
The STV rule eliminates one candidate per round.  Initially,
$\activecandidates=\cC$.  In each round, we count first-place votes in the
restricted profile
$(\pi_1|_{\activecandidates},
\ldots,
\pi_n|_{\activecandidates})$ and eliminate a candidate with the fewest such
votes.  After
$m-1$ eliminations, the remaining candidate wins.
\end{itemize}

In all rules above, we assume that ties are broken based on some prespecified order on the candidates. For our work, a specific order is unimportant, as long as it is consistent.

\subsubsection{Margin of Victory.}

\label{subsubsect:MoV}

Margin of victory (e.g. \cite{Xia12}), a standard measure for decisiveness (or stability) of an election, will play a central role in our work. Specifically, the margin of victory is the number of voters required to change the outcome of an election, as defined more formally below.

%We use the standard definition of margin of victory from \cite{Xia12},

\begin{definition}[Margin of Victory (MoV)]
Given a preference profile $\bpi$ and a voting rule $r$, the margin of
victory, denoted by $\mov(\bpi,r)$, is the smallest integer $k$ for which
there exists a $k$-neighbor $\bpi'$ of $\bpi$ such that
$r(\bpi') \ne r(\bpi)$.
\end{definition}

%Theorems~5 and~7 of 
\citet{Xia12} show that $\mov(\bpi,r)$ can be computed in
polynomial time when $r$ is Plurality, Condorcet, a scoring rule, or Plurality with Runoff.
In contrast, it is known that computing $\mov$ is NP-hard for Maximin and STV~\cite{Xia12,BO91}.
%In contrast, computing the margin of victory for Maximin is NP-complete
%\cite[Theorem~3]{Xia12}. A similar NP-hardness is known for STV~\cite{BO91}.  
%Nevertheless, \cite[Theorem~12]{Xia12} gives a
%constant-factor certificate: if $w$ and $w'$ have the largest and second-largest
%maximin scores, respectively, and $r$ is Maximin, then, writing
%$\Delta_{\mathrm{mm}}
%:= \mathsf{score}_{\mathrm{maximin}}(w)
%   - \mathsf{score}_{\mathrm{maximin}}(w')$.
%It is shown that 
%    \begin{equation} \frac{\mov(\bpi,r)}{2}
%    \leq \frac{\Delta_{\mathrm{mm}}}{4}
%    \leq \mov(\bpi,r). \label{eqn:2approx}\end{equation}
%For STV, it is NP-complete to distinguish $\mov(\bpi,r)=1$ from
%$\mov(\bpi,r)=2$~\cite{BO91}.  Thus, computing $\mov(\bpi,r)$ exactly for STV is
%also intractable.

As explained below, differential privacy (DP) requires the output to be random. We desire algorithms that output the correct winner with high probability as long as $\mov$ is large. This is captured in the usefulness definition below.
%We define the usefulness in the following definition.

\begin{definition}[$(\tau,\gamma)$-usefulness]
We say that an algorithm $\alg$ is $(\tau, \gamma)$-useful for voting rule
$r$ if, for all $\bpi$ such that $\mov(\bpi, r) \geq \tau$, we have
$\Pr[\alg(\bpi) = r(\bpi)] \geq 1 - \gamma$.

We also refer to $\tau$ as the ($\gamma$-)decisive margin for $\cM$.
\end{definition}

\subsection{Differential Privacy}

We next recall the definition of (central) differential privacy, adapted to our voting setting.

\begin{definition}[DP \cite{DworkMNS06,DworkKMMN06}] \label{def:central-dp}
An algorithm $\alg$ is 
%$(\eps,  \delta)$-DP 
$\eps$-DP
if, for all neighboring profiles
$\bpi,\bpi'$ and all output events $\cO$,
$\Pr[\alg(\bpi) \in \cO] \leq e^{\eps} \cdot \Pr[\alg(\bpi') \in \cO]$. %$ + \delta$.
%When $\delta=0$, we write $\eps$-DP.
\end{definition}

One of the most widely-used tools in DP is the exponential mechanism~\cite{McSherryT07}, which allows us to optimize any score function. The mechanism description and its utility are summarized below.

%We define the exponential mechanism of \cite{McSherryT07} below.

\begin{definition}[Exponential mechanism]\label{def:expo-mech}
Let $\cO$ be a finite set, and let
$q:\Pi_{\cC}^n\!\times\! \cO \to \mathbb{R}$ be a score function.
Define the sensitivity of $q$ as
$\Delta(q) \ :=\ \max_{\bpi\sim\bpi',\,a\in \cO}
\bigl|q(\bpi,a)-q(\bpi',a)\bigr|$ where $\bpi\sim\bpi'$ denotes neighboring profiles.
For $\eps>0$, the exponential mechanism $\alg_{\exp}$ outputs an element
$a\in \cO$ with probability %proportional to
$\Pr[\alg_{\exp}(\bpi)=a]\ \propto\ \exp\!\left(\frac{\eps\cdot q(\bpi,a)}{2\,\Delta(q)}\right)$.
Then $\alg_{\exp}$ is $\eps$-DP.
\end{definition}

%The utility of this mechanism is summarized below.

%We note that the exponential mechanism applies only when the output set $\cO$ is finite, and it requires evaluating the score function $q(\bpi,a)$ for all $a \in \cO$.

\begin{theorem}
%[{\cite[Theorem~3.11]{DworkRoth14}}]
%[Exponential mechanism utility]
\label{thm:expmech-utility}
%For the exponential mechanism,  gives the following guarantee.
%Let $\cO$ be a finite range of outputs. 
Define
$\OPT(\bpi)\ :=\ \max_{a\in\cO} q(\bpi,a)$.
Then, for any $\beta\in(0,1)$, with probability at least $1-\beta$, $q(\bpi,\alg_{\exp}(\bpi)) \ge \OPT(\bpi) - \frac{2\Delta(q)}{\eps}\ln\!\left(\frac{|\cO|}{\beta}\right)$.
\end{theorem}

\subsection{Local Differential Privacy}

We use the standard local model of differential privacy
\cite{kasiviswanathan2011can}.  In the non-interactive model, each voter applies a private
randomizer on their own data and send the results to the aggregator.

\begin{definition}[Non-interactive $\eps$-LDP] \label{def:non-interactive-local-dp}
A local randomizer $\mathcal{R}$ satisfies $\eps$-LDP if, for all rankings
$\pi,\pi'$ and all output events $\cO$,
$\Pr[\mathcal{R}(\pi) \in \cO]
\leq e^{\eps}\Pr[\mathcal{R}(\pi') \in \cO]$.
A non-interactive $\eps$-LDP algorithm sends each voter a randomizer fixed
before any messages are observed and receives one message from each voter.  In other words, it
can be written $\alg(\bpi)=\mathcal{A}(\mathcal{R}_1(\pi_1),\ldots,
    \mathcal{R}_n(\pi_n))$,
where each $\mathcal{R}_i$ is an $\eps$-LDP randomizer and $\mathcal{A}$ is
the aggregator.
\end{definition}

On the other hand, the interactive setting allows the voters and the aggregator to perform arbitrary communication protocol. The privacy constraint here is that the transcript of the protocol (which includes all messages from all the parties) has to be differentially private.

\begin{definition}[Interactive $\eps$-LDP] \label{def:interactive-local-dp}
An interactive local protocol accesses the profile only by querying voters.  In
each round, as a function of the transcript so far, it chooses a set of voters and
randomizers to apply to those voters' rankings.  The transcript is the
ordered record of the public randomness, queried voters, selected channels, and
exchanged messages.
%; private randomness appears only through the messages it generates.  
Let $\mathsf{Tr}(\bpi)$ denote the full transcript on profile
$\bpi$.  The protocol
is $\eps$-LDP if, for every voter $i$, every pair of profiles
$\bpi,\bpi'$ differing only in voter $i$'s ranking, and every transcript event
$\cO$, $\Pr[\mathsf{Tr}(\bpi)\in\cO]
    \le e^\eps\Pr[\mathsf{Tr}(\bpi')\in\cO]$.
\end{definition}

%\pasin{The above definition doesn't make sense for interactive algorithms. Need to be fixed.}
%\vorapong{Modified. Please let me know if this version works.} \pasin{I don't think so: The way the local randomizer is defined, it only takes in $\pi$ as the input and works in one step, this is only for non-interactive LDP. For interactive LDP, the randomizers can be more complicated (e.g. can take in the transcript). Maybe it's best to just define the two models separately, and just use the ``protocol'' definition for interactive.} \vorapong{Here is my third trial. I am wondering if this separated definition will make it slightly hard for Section 5 to link with this new definition, but I will try to do that later.}

For the vector releases in Sections~\ref{sec:noninteractive-ldp} and
\ref{sec:interactive-ldp}, we use the
LDP mean-estimation randomizer of~\citet{DuchiJW18}.
%We record the property of their construction that we need.  
We write
$\kappa_\varepsilon=(e^\varepsilon+1)/(e^\varepsilon-1)$, which is $\Theta\Paren{\max\left\{\frac{1}{\eps}, 1\right\}}$ for all $\eps > 0$.
%$O(1/\varepsilon)$ for $0<\varepsilon\le1$, and $O(1)$ otherwise.

\begin{lemma}[\citealt{DuchiJW18}]
\label{lem:djw-randomizer}
There is a constant $C_{\mathrm{DJW}}$ with the following property.
For every $x\in\mathbb{R}^d$ satisfying $\|x\|_2\le r$, the DJW randomizer
$\mathcal{R}_{\mathrm{DJW}}(x;r,\varepsilon)$ is $\varepsilon$-LDP, and
is unbiased, i.e.
    $\mathbb{E}[\mathcal{R}_{\mathrm{DJW}}(x;r,\varepsilon)]=x$.
Moreover, for every unit vector $u$, the centered projection
    $\left\langle
    \mathcal{R}_{\mathrm{DJW}}(x;r,\varepsilon)-x,u
    \right\rangle$
is sub-Gaussian with parameter at most
$C_{\mathrm{DJW}}\kappa_\varepsilon r$.
\end{lemma}

%The construction first maps $x$ unbiasedly to the radius-$r$ sphere and then
%reports a suitably scaled random point from one of the two hemispheres, choosing
%between the hemispheres with odds $e^\varepsilon$.  The likelihood ratio is
%therefore at most $e^\varepsilon$, while spherical concentration gives the
%sub-Gaussian guarantee in Lemma~\ref{lem:djw-randomizer}.

\section{Central Differential Privacy}
\label{sec:central-dp}

In this section, we describe our algorithms and lower bounds for the central DP model (Definition~\ref{def:central-dp}).

\subsection{Algorithms}

\paragraph{Generic algorithm.} We start by describing a mechanism that works
for any voting rule~$r$.

Our mechanism applies the exponential mechanism
(Definition~\ref{def:expo-mech}) with
$\cO=\cC\cup\{\bot\}$ and score function
\begin{equation}
q(\bpi,a)\;:=\;
\begin{cases}
0, & \text{if } a = r(\bpi),\\
-\mov(\bpi,r), & \text{otherwise.}
\end{cases} \label{eqn:score}
\end{equation}
For neighboring profiles $\bpi$ and $\bpi'$, we have
$|\mov(\bpi,r)-\mov(\bpi',r)|\le1$.  Moreover, $r(\bpi)\ne r(\bpi')$ can occur
only when $\mov(\bpi,r)\le1$.  Therefore, $\Delta(q)\leq 1$.
Thus, a direct application of Theorem~\ref{thm:expmech-utility} yields the following.

\begin{theorem} \label{thm:central-em}
For any $\eps > 0$, there exists an $\eps$-DP algorithm that is
$\left(\frac{2 \log((m+1)/\gamma)}{\eps} + 1, \gamma\right)$-useful for all
$\gamma \in (0, 1/2)$. \label{thm:cdp-upper}
\end{theorem}

The theorem is only information-theoretic: it is efficient only when the MoV score can
be evaluated efficiently.  As discussed in Section~\ref{subsubsect:MoV}, this is
the case for Plurality, Condorcet, scoring rules, and Plurality with Runoff.

\paragraph{An efficient algorithm for Maximin.}

For Maximin, exact MoV computation is NP-complete.  Nevertheless, the
approximation of~\cite[Theorem~12]{Xia12} gives a useful surrogate.  If $w$
and $w'$ have the largest and second-largest maximin scores, respectively, and
$r$ is Maximin, write
$\Delta_{\mathrm{mm}}
:= \mathsf{score}_{\mathrm{maximin}}(w)
   - \mathsf{score}_{\mathrm{maximin}}(w')$.
Accounting for our fixed tie-breaking rule adds at most one vote to the usual
two-approximation guarantee, so
    \begin{equation} \frac{\mov(\bpi,r)-1}{2}
    \leq \frac{\Delta_{\mathrm{mm}}}{4}
    \leq \mov(\bpi,r). \label{eqn:2approx}\end{equation}
%Instead, we use the
%constant-factor surrogate from~\eqref{eqn:2approx}.  Let $w$ and $w'$ be
%candidates with the largest and second-largest maximin scores, respectively, and
%write $\Delta_{\mathrm{mm}}$ for their score gap.  
Importantly, this surrogate $\Delta_{\mathrm{mm}}$ can be efficiently computed. Thus, we may apply the exponential mechanism to the surrogate score, defined below, instead.
\begin{equation}
q'(\bpi,a)\;=\;
\begin{cases}
0, & \text{if } a = r(\bpi),\\
-\Delta_{\mathrm{mm}}/4, & \text{otherwise.}
\end{cases} \label{eqn:score-maximin}
\end{equation}

Changing one vote changes every maximin score by at most $2$, and hence changes
$\Delta_{\mathrm{mm}}$ by at most $4$ while the winner remains fixed.  If the
winner changes, both the old and new top-two gaps are at most $4$, since the old
and new winners' scores can move by at most $2$ each.  It follows in both cases that
$\Delta(q')\leq1$.  Thus, as above, the exponential mechanism yields an
efficient algorithm for Maximin:

\begin{theorem}
For any $\eps > 0$, there exists a polynomial-time $\eps$-DP algorithm for
Maximin that is
$\left(\frac{4 \log(m/\gamma)}{\eps} + 2, \gamma\right)$-useful for all
$\gamma \in (0, 1/2)$. \label{thm:cdp-upper-maximin}
\end{theorem}

\paragraph{Computational barrier for STV.}
For STV, the MoV is NP-hard to compute. Thus, we cannot directly instantiate \Cref{thm:central-em} to get an efficient algorithm. One might hope--via some surrogate score similar to Maximin or some other means--that we can devise an efficient algorithm for STV. Unfortunately, we show that this is impossible under a standard complexity assumption, as stated more precisely below.

%For STV, the generic exponential-mechanism algorithm need not be efficient
%because the MoV score is hard to approximate.  The next theorem strengthens this
%observation: under a standard complexity assumption, no polynomial-time central
%DP algorithm can simultaneously have logarithmic-MoV usefulness and very high
%confidence on high-MoV instances.  The required guarantees are weak enough that
%the inefficient exponential mechanism satisfies them.

%On the other hand, we show that this is impossible for STV, as formalized below.

\begin{theorem} \label{thm:stv-no-efficient-dp}
Suppose that NP $\nsubseteq$ BPP.  For any positive constants
$C_1,C_2,C_3,\eps$ and $\gamma\in(0,1)$ with $C_3>4C_1$, there is no $\eps$-DP
$(nm)^{O(1)}$-time algorithm that is simultaneously
$\left(C_1\log(m)/\eps,\gamma\right)$-useful and
$\left(m^{C_2},1/m^{C_3}\right)$-useful.
\end{theorem}

As stated in the introduction, this is one of the few examples of computational problems for which DP makes it intractable. % Furthermore, the connection between hardness of approximation results and DP  
Below we sketch the proof of this theorem. 

\paragraph{Hardness of approximation of MoV.} As a first step, we show that it is NP-hard to approximate $\mov$ of STV to within a factor of $n^{1 - \zeta}$ for any constant $\zeta > 0$. This factor is (essentially) optimal since a trivial $n$-approximation exists. Moreover, our result is a direct improvement over that of \citet{BO91} who proved an NP-hardness of approximation with a factor of 2. %We provide a formal statement and further discussion of our result in \Cref{sec:non-private-STV}.

\paragraph{Computational hardness for DP algorithm.} Our inapproximability result seems ostensibly unrelated to privacy. Nevertheless, we use this to prove \Cref{thm:stv-no-efficient-dp} as follows.

Assume for the sake of contradiction that an algorithm $\alg$ in \Cref{thm:stv-no-efficient-dp} exists. We use it to estimate the MoV of STV as follows by running $\alg(\bpi,r)$
many times. If every run returns the original winner, then we know that the MoV should be large. Meanwhile, if at least one run returns a different winner, then we know that the MoV should be small. To make this proof work, we need a rather precise set of parameters, which we specify in the appendix.

% The
%full proof appears in Appendix~\ref{app:missing-proof-section-3}.

%Theorem~\ref{thm:stv-no-efficient-dp} gives a natural task that is easy without
%privacy (output $r(\bpi)$) and possible with inefficient privacy
%(Theorem~\ref{thm:cdp-upper}), but unlikely to admit both efficiency and privacy
%with the desired usefulness.  Related efficiency/privacy separations are known
%for private synthetic data generation~\cite{dwork2009complexity,UllmanV20}.

\subsection{Lower Bounds}
We now prove matching central-DP lower bounds using a packing argument
\cite{HardtT10}. Such an argument requires many nearby profiles with different
winners, as formalized below.

\begin{assumption} \label{assumption:packing}
There exists a constant $\kappa>1$ such that, for every $t\in\mathbb{N}$, there
are profiles $\bpi_1,
\ldots,
\bpi_m$ satisfying:
\begin{enumerate}
\item $r(\bpi_i) = c_i$ for all $i \in [m]$, \label{prop:diff-winners}
\item $\mov(\bpi_i, r) \geq t$, \label{prop:label}
\item $\bpi_i, \bpi_{i'}$ are $(\kappa t)$-neighbors for all $i, i' \in [m]$.
\end{enumerate}
\end{assumption}

Under the above assumption, the standard packing lower bound proof \cite{HardtT10} yields the following:

\begin{lemma} \label{lem:packing-generic}
Let $\eps > 0$ and $\gamma \in (0, 0.2)$.
Suppose that $r$ satisfies Assumption~\ref{assumption:packing}. Then, any $\eps$-DP algorithm $\alg$ that is $(\tau, \gamma)$-useful for rule $r$ must satisfy
$\tau \geq \left\lfloor \frac{\ln(0.4m/\gamma)}{\kappa \eps} \right\rfloor$.
\end{lemma}

Note here that, when $\kappa$ is a constant, this asymptotically matches our upper bound (\Cref{thm:central-em}). For the remainder of this section, we show that every rule $r$ considered in our work satisfies Assumption~\ref{assumption:packing} for a constant $\kappa$.

%The proof is deferred to Appendix~\ref{app:missing-proof-section-3}.

\subsubsection{Rules satisfying the majority criterion.}

A voting rule satisfies the \emph{majority criterion} if any candidate ranked
first by more than half of the voters must win. By simply taking each $\pi^i$ to be a profile with $2t$ voters all of whom rank
$c_i$ first, we can immediately show:

\begin{lemma} \label{lem:packing-majority}
Any voting rule that satisfies Majority Criterion also satisfies Assumption~\ref{assumption:packing} with $\kappa = 2$.
\end{lemma}

Plurality, Condorcet, Maximin, Plurality with Runoff, and STV all satisfy the majority
criterion, so Lemmas~\ref{lem:packing-generic} and~\ref{lem:packing-majority}
imply the following asymptotically tight lower bound.

\begin{corollary}
\label{cor:packing-majority-lower}
For any $\eps > 0$ and $\gamma \in (0, 0.2)$, any $\eps$-DP algorithm $\alg$ that is $(\tau, \gamma)$-useful for Plurality, Condorcet, Maximin, Plurality with Runoff, or Single Transferable Vote must satisfy $\tau = \Omega(\ln(m/\gamma) / \eps)$.
\end{corollary}

\subsubsection{Scoring-based rules.}

Some scoring rules, such as Veto
%\footnote{Veto is the voting rule with scoring vector $(1, \dots, 1, 0)$.}
, do not satisfy the majority criterion, so the
previous construction does not apply.  We instead build a %symmetric 
profile with
many voters and then carefully perturb it to favor any prescribed candidate. By such a construction, we can show:

\begin{lemma} \label{lem:packing-scoring-based}
Any non-trivial\footnote{We say that a scoring rule is non-trivial if the scoring vector $\bs = (s_1, \dots, s_m)$ satisfies $s_1 \ne s_m$.} scoring rule satisfies Assumption~\ref{assumption:packing} with $\kappa = 4$.
\end{lemma}

Combining Lemmas~\ref{lem:packing-generic},~\ref{lem:packing-scoring-based}
yields the tight lower bound.

\begin{corollary}
\label{cor:packing-scoring-lower}
For any $\eps > 0$ and $\gamma \in (0, 0.2)$, any $\eps$-DP algorithm $\alg$ that is $(\tau, \gamma)$-useful for any non-trivial scoring rule must satisfy $\tau \geq \Omega(\ln(m/\gamma) / \eps)$.
\end{corollary}

%\pasin{I feel like the stuff below can be moved to the appendix since it doesn't fit with the rest?}
%\vorapong{I supposed you mean the following subsubsection. I kept it because it sounds interesting, but I will not mind moving it to Appendix.} \pasin{I'm actually fine with keeping this, but the text needs to put this into context of the high-level picture. Currently, this sounds like a standalone thing and it's pretty confusing what this subsection is trying to accomplish. (There is something interesting going on here. It just isn't connected well to the other results.)}
%vorapong{I tried to add a paragraph to link with the previous result. However, after I have added the paragraph, I agree that this discussion might be moved to appendix if we run out of space.}

\subsection{Rules Beyond Our Framework}

While our framework above provides tight bounds of $\Theta\Paren{\frac{\log(m/\gamma)}{\eps}}$ for all rules of interest in this paper, we remark that there are in fact rules whose optimal decisive margin is not $\Theta\Paren{\frac{\log(m/\gamma)}{\eps}}$. One such rule is the following
%
%\subsubsection{A rule outside the packing framework.}
%
%The preceding results establish the lower bound of $\Omega(\log m/\eps)$ using the packing property stated in Assumption~\ref{assumption:packing}. Although this property is satisfied by the standard voting rules considered in this paper, it does not necessarily hold for every voting rule. The following example demonstrates that, in the absence of such a packing, the $\Omega(\log m/\eps)$ lower bound may fail. Consider a 
``near unanimity'' rule
parameterized by $n_0,\tau_0$ with $n_0>2\tau_0$:
\begin{itemize}
\item When $n \leq n_0$, we always declare a tie, i.e., $r(\bpi) = \bot$.
\item When $n > n_0$, the winner is the candidate that is most preferred by at least $n - \tau_0$ voters. If no such candidate exists, declare a tie.
\end{itemize}
%For any $d$ satisfying $d \le n-2\tau_0-t$, one can verify the following.
%Let $\bpi_i$ be any profile with $r(\bpi_i)=c_i$ and $\mov(\bpi_i,r)\ge t$.
%Then for every profile $\bpi'$ that is a $d$-neighbor of $\bpi_i$, we must have $r(\bpi') \in \{c_i,\bot\}$.

%When $n_0$ is large and $\tau_0,t$ is small, any two profiles $\bpi_i,\bpi_{i'}$
%with $r(\bpi_i)=c_i$ and $r(\bpi_{i'})=c_{i'}$ (where $i\neq i'$) are separated by a distance
%far exceeding $\mov(\bpi_i,r)$. In particular, there is no constant $\kappa$ for which
%$\bpi_i$ and $\bpi_{i'}$ can be $(\kappa t)$-neighbors for such $t$.
%Therefore, this voting rule does not %satisfy %Assumption~\ref{assumption:packing}.

%Indeed, this 
This
rule admits an $\eps$-DP algorithm with
$(o_{\gamma,\eps}(\log m),\gamma)$-usefulness.  The algorithm privately counts
first-place votes using discrete Laplace noise and then applies the rule.  For
sufficiently large $n_0$ and $\tau_0\ge\Theta(\log(m)/\eps)$, the relevant
scores are far from the decision boundary except in a constant-width band,
yielding $(O_{\gamma,\eps}(1),\gamma)$-usefulness.
\section{Non-Interactive Local Differential Privacy}
\label{sec:noninteractive-ldp}

We next consider non-interactive LDP (Definition~\ref{def:non-interactive-local-dp}). We remark that all of our (interactive or non-interactive) LDP algorithms are efficient, and we will not state this explicitly in the theorem statements.

\subsection{Plurality}

We start by the Plurality rule, which has the simplest algorithm and lower bound proof.

\label{subsubsec:alg-plurality-voting}

\paragraph{Upper bound.} 
The Plurality algorithm is simply the LDP histogram algorithm\footnote{We state our algorithm using the randomizer of \cite{DuchiJW18} to unify the notations across all the results. However, other LDP histogram algorithms, e.g. RAPPOR \cite{ErlingssonPK14}, also work here.}. More specifically,
%, given in Algorithm~\ref{alg:publish-plurality-scores} in Appendix~\ref{app:plurality-details},
each voter privatizes the one-hot vector of their
favorite candidate using the DJW randomizer.  That is, if voter $i$ ranks $w_i$ first, let
$f(\pi_i)=v_i\in\{0,1\}^m$ where $v_{i,w_i}=1$ and $v_{i,c}=0$ for every
$c\ne w_i$.  Since $\|v_i\|_2=1$, the voter can release
$\tilde{v}_i=\mathcal{R}_{\mathrm{DJW}}(v_i;1,\varepsilon)$ while satisfying $\eps$-LDP.
The aggregator sums the noisy vectors and outputs
$\tilde{w} \in \arg\max_{c_j \in \cC} \sum_{i \in [n]} \tilde{v}_{i,j}$. The following theorem follows from the sub-Gaussian guarantee in Lemma~\ref{lem:djw-randomizer}.

%The proof of the following theorem is given in Appendix~\ref{app:plurality-details} and follows from the sub-Gaussian guarantee in Lemma~\ref{lem:djw-randomizer}.

\begin{theorem}
    \label{thm:ldp-plurality-upper}
    There is a non-interactive  $\eps$-LDP algorithm for Plurality which is
        $\left(
        %2+C_{\mathrm{DJW}}\kappa_\varepsilon
        %\sqrt{2n\log(2m/\gamma)},
        O\left(\max\left\{\frac{1}{\eps}, 1\right\} \cdot \sqrt{n \log(m/\gamma)}\right),
        \gamma
        \right)$-useful
    for $\eps>0$ and $\gamma\in(0,1/2)$.
\end{theorem}

\paragraph{Lower bound.} A direct two-point mean-testing reduction gives the
following non-interactive LDP lower bound, which asymptotically matches the upper bound for $\gamma < 1/m^{\Omega(1)}$ and $\eps \leq \Theta(1)$.

\begin{theorem}
    \label{thm:ldp-plurality-lower}
    There are constants $C',c,\varepsilon_0>0$ such that the
    following holds for Plurality.  For $0<\varepsilon\le\varepsilon_0$ and
        $e^{-cn\varepsilon^2}
        \le\gamma<\frac{1}{10}$,
    no non-interactive $\varepsilon$-LDP algorithm is
        $\left(\frac{C'}{\varepsilon}
        \sqrt{n\log(1/\gamma)},
        \gamma
        \right)$-useful.
\end{theorem}

\paragraph{Proof sketch.}
Draw every vote independently from one of two binary distributions whose
means are
$\pm\Theta(\sqrt{\log(1/\gamma)/(n\varepsilon^2)})$.  With probability at
least $1-\gamma$, the resulting Plurality election has the corresponding
winner and margin
$\Omega(\sqrt{n\log(1/\gamma)}/\varepsilon)$.  A mechanism useful below this
threshold would therefore distinguish the two distributions with error
$O(\gamma)$.  The strong data-processing inequality for local
privacy~\cite{DuchiJW18} and the Bretagnolle--Huber testing bound rule this out. \hfill $\square$

\subsection{Condorcet, Maximin, Plurality with Runoff}
\label{subsec:main-nonint-pairwise}

Next, we consider Condorcet, Maximin, Plurality with Runoff. For these rules, the decisive margin increases (compared to Pluraity) increases by a factor of $O(\sqrt{m})$. This increase is due to the fact that each voter now needs to privatize the pairwise election, which increases the $\ell_2$ norm of the contribution vector to $O(\sqrt{m})$, as formalized below.

\subsubsection{Upper bounds.}

Fix an orientation of the $q=\binom{m}{2}$ unordered candidate pairs.  For a
ranking $\pi$, define its pairwise vector $P(\pi)\in\{-1,1\}^q$ by
    $P_{a,b}(\pi)
    =\mathbf{1}[\pi(a)\succ\pi(b)]
    -\mathbf{1}[\pi(b)\succ\pi(a)]$
for each oriented pair $(a,b)$.  Thus $\|P(\pi)\|_2=\sqrt q$, and the sum of
the $(a,b)$ coordinates over all voters is exactly the pairwise margin
$D_{\bpi}(a,b)$.

Each voter releases
    $Z_i=\mathcal{R}_{\mathrm{DJW}}(P(\pi_i);\sqrt q,\varepsilon)$,
using %the randomizer from 
Lemma~\ref{lem:djw-randomizer}.  This is a
$\varepsilon$-LDP release.  The aggregator sets
$\widetilde P=\sum_{i=1}^n Z_i$ and reads $\widetilde D(a,b)$ from the
appropriately signed coordinate of $\widetilde P$.  The estimate is unbiased,
and independence across voters together with Lemma~\ref{lem:djw-randomizer}
gives
    $\Pr\!\left[
    |\widetilde D(a,b)-D_{\bpi}(a,b)|>t
    \right]
    \le 2\exp\!\left(
    -\frac{t^2}{2C_{\mathrm{DJW}}^2\kappa_\varepsilon^2qn}
    \right)$.

In the runoff algorithm we also use the DJW-privatized plurality vector
$\tilde{v}=(\tilde{v}_i)_{i\in[m]}$ from
Section~\ref{subsubsec:alg-plurality-voting}, with privacy budget
$\varepsilon/2$; the DJW pairwise release also uses budget $\varepsilon/2$.
By composition, their joint release is  $\varepsilon$-LDP.

\textbf{Condorcet:} The aggregator outputs a candidate $c$ if
$\widetilde{D}(c,a)>0$ for every $a\in\cC\setminus\{c\}$, and outputs $\bot$ if
no such candidate exists.

\textbf{Maximin:} The aggregator computes each candidate's noisy maximin
score
    $\widetilde{\score}_{\mathrm{maximin}}(a) := \min_{b \in \cC \setminus \{a\}}
    \widetilde{D}(a, b),$
and outputs the candidate with maximum score. %$\hat{c}\in\arg\max_{a\in\cC}\widetilde{\score}_{\mathrm{maximin}}(a)$.

\textbf{Plurality with Runoff:} Based on the vector $\tilde{v}$, the aggregator selects the top two candidates $c, c'$ with the highest noisy plurality scores.
For the second stage of the runoff rule, the aggregator computes the noisy pairwise margin $\widetilde{D}(c, c')$
and outputs $c$ if $\widetilde{D}(c, c') > 0$, and $c'$ otherwise.

%The pseudocode for Plurality with Runoff and the detailed usefulness statements
%and proofs for Condorcet, Maximin, and Plurality with Runoff are deferred to
%Appendix~\ref{app:cmr-details}
%(Algorithm~\ref{alg:noninteractive-plu-runoff} and
%Theorems~\ref{thm:ldp-condorcet-upper}--\ref{thm:ldp-runoff-upper}).
Thus, we obtain the following theorem from the sub-Gaussian guaranty in Lemma \ref{lem:djw-randomizer}. 

\begin{theorem} \label{thm:non-int-ldp-algo-pairwise}
The non-interactive algorithms described above for Condorcet, Maximin, and
Plurality with Runoff are $\varepsilon$-LDP and
$(\tau,\gamma)$-useful for
    $\tau = \tilde{O}_{\varepsilon,\gamma}(m\sqrt{n})$.
\end{theorem}
\paragraph{Lower bounds.}
To prove (nearly) matching lower bounds, we reduce from the two-way marginal problem:
\begin{definition}[$k$-way marginal]
\label{def:k-way-marginal}
For a Boolean dataset $X=(x_1,\ldots,x_N)$ where $x_1, \dots, x_N \in \{0, 1\}^d$ and
$A\subseteq[d]$ with $|A|=k$, the $k$-way marginal indexed by $A$ is the
conjunction count $Q_A(X)=\sum_{i=1}^N\prod_{j\in A}x_{i,j}$.
\end{definition}

Our lower bound reduces non-interactive LDP release of
\emph{all} normalized two-way marginals to private winner recovery.  A key point in our reduction for non-interactive algorithms is that we can \emph{reuse privatized values} to answer multiple queries\footnote{This observation has been used before in non-private LDP lower bounds, e.g. in \cite{ghazi2023computing,eden2025triangle}.}. At a high-level, this framework proceeds as follows\footnote{This description is simplified for the main body. E.g., in the actual reduction, we use \emph{two} data-dependent rankings per one $x_i$.}:
\begin{itemize}
\item For each $x_i$, we create one \emph{data-dependent} ranking $\pi_i$ which is then privatized to $y_i = \cR_i(\pi_i)$.
\item For each query $A$, do the following (possibly repeatedly):
\begin{itemize}
\item We create additional \emph{query-dependent} rankings $\pi'_1, \dots, \pi'_L$ and run the randomizer to get $y'_j = \cR_j(\pi_j)$
\item From the combined output $y_1, \dots, y_N, y'_1,\dots,y'_L$, we compute the winner of the election.
\item Use the winner to estimate $Q_A(X)$.
%\item This process is repeated potentially multiple times to estimate $Q_A(X)$.
\end{itemize}
\end{itemize}

As mentioned earlier, we reuse the randomized responses $y_1, \dots, y_N$ multiple times to answer multiple queries.

We defer the exact descriptions of the rankings to the appendix, but we sketch the high-level idea here. Roughly speaking, we think of each
candidate $c_j$ as a coordinate $j$ of the dataset. Our \emph{data-dependent} rankings are constructed 
%For a query
%$A=\{j,k\}$, our construction adds two
%data-dependent voters per user
so that the head-to-head margin between $c_j$ and
$c_k$ encodes $\sum_i x_{i,j}x_{i,k}$. The \emph{query-dependent} voters act as a ``thresholding query'' that checks whether $\sum_i x_{i,j}x_{i,k} \geq \tau$ for some threshold $\tau$. If this is true, then $c_j$ is the winner; otherwise, $c_k$ is the winner. % so
%that $c_j$ is the Condorcet or Maximin winner, or the runoff winner against
%$c_k$, exactly
%when this marginal is above the tested threshold.  
The margin of victory in the
constructed election is proportional to the distance from the threshold.  A
nonadaptive scan of all fixed thresholds therefore recovers all two-way
marginals under LDP.  Combining this reduction with the known two-way marginal lower
bound (in Theorem~\ref{thm:edmonds-local-marginal-lower}) yields the following
winner-recovery lower bound.

\begin{corollary}
\label{cor:ldp-cmr-main-lower}
There are constants $c,\varepsilon_0>0$ such that, for
$0<\varepsilon\le\varepsilon_0$, every non-interactive $\varepsilon$-LDP
algorithm for Condorcet, Maximin, or Plurality with Runoff that is
$(\tau, \gamma)$-useful with $\gamma \le c/(m^2n^2)$ must satisfy
    $\tau=\widetilde{\Omega}\!\left(
    \min\left\{n,\frac{m\sqrt{n}}{\varepsilon}\right\}\right)$.
\end{corollary}

\subsection{STV}
\paragraph{Upper bound.} For STV, the aggregator must know noisy plurality scores after every possible
set of eliminations. Therefore, the $\ell_2$ norm of each contribution vector becomes $O(\sqrt{2^m}) = O(2^{m/2})$, which results in such a blow-up factor in the decisive margin.

More formally, let
    $\mathcal{I} := \{(S,c): S\subseteq \cC,\ |S|\ge2,\ c\in S\}$
and define $T(\pi)\in\{0,1\}^{\mathcal{I}}$ by setting
$T_{S,c}(\pi)=1$ if $c$ is the highest-ranked candidate in $S$ under $\pi$,
and $T_{S,c}(\pi)=0$ otherwise.
For each non-singleton set $S$, exactly one coordinate indexed by $S$ equals
one.  Hence
    $\|T(\pi)\|_2=r_T:=\sqrt{2^m-m-1}$.
Each voter releases
    $Z_i=\mathcal{R}_{\mathrm{DJW}}(T(\pi_i);r_T,\varepsilon)$.
By Lemma~\ref{lem:djw-randomizer}, this release is
$\varepsilon$-LDP, unbiased, and has coordinatewise sub-Gaussian error
with parameter at most $C_{\mathrm{DJW}}\kappa_\varepsilon r_T$.

Our algorithm then simply follows the STV protocol, eliminating candidates one-by-one, where we use the above release to estimate the score of every remaining candidate.

%The pseudocode and detailed STV usefulness statement and proof are deferred to
%Appendix~\ref{app:stv-details} (Algorithm~\ref{alg:noninteractive-stv} and
%Theorem~\ref{thm:ldp-stv-upper}).

\begin{theorem} \label{cor:ldp-stv}
The non-interactive algorithm described above for STV is $\varepsilon$-LDP and $(\tau,\gamma)$-useful for
    $\tau = \tilde{O}\Paren{\max\left\{\frac{1}{\eps}, 1\right\} \cdot 2^{m/2}\sqrt{n \log(1/\gamma)}}$.
\end{theorem}

\paragraph{Lower bound.} 
Similar to above, we reduce from $k$-way marginal. We can now utilize a larger range of values of $k$, since the query-dependent rankings can be used to ``select'' the candidates eliminated in the first rounds. Specifically, we instantiate two extreme values of $k$ below: $k = 2$ and 
%We construct, for every
%queried set $A$ of $k$ coordinates, an STV election in which an auxiliary
%candidate $c^*$ wins exactly when the corresponding $k$-way marginal exceeds
%the tested threshold.  The case $k=2$ yields the two-way lower bound, while
$k=\Theta(m)$.
The former yields a better dependency for large $m$, whereas the latter
gives a better %candidate 
bound for small $m$.  %Combining the two
%choices gives the following corollary.

\begin{corollary}
\label{cor:main-ldp-stv-kway-edmonds}
There are constants $c,\varepsilon_0>0$ such that, for
$0<\varepsilon\le\varepsilon_0$, every non-interactive
$\varepsilon$-LDP algorithm for Single Transferable Vote that is
$(\tau, \gamma)$-useful with $\gamma \le c/(m^{m/2}n^2)$ must satisfy $\tau=\widetilde{\Omega}\!\left(
      \max\{L_2,L_{\mathrm{high}}\}\right)$ where
\[
\begin{aligned}
    L_2=\min\left\{n,\frac{m\sqrt n}{\varepsilon}\right\},
    L_{\mathrm{high}} =\min\left\{\frac{n}{m},
      \frac{\sqrt n}{\varepsilon}2^{\Omega(\sqrt m)}\right\}.
\end{aligned}
\]
\end{corollary}

We remark that, while we require the failure probability $\gamma$ to be quite small, the dependency on $\gamma$ in the upper bound in \Cref{cor:ldp-stv} is only $O(\sqrt{\log \gamma})$. Therefore, this only effects our upper bound by a factor of $\tilde{O}(\sqrt{m})$. On the other hand, the more interesting remaining gap here is the $2^{O(m)}$ term versus the $2^{\Omega(\sqrt{m})}$ term in our upper and lower bounds, respectively. Closing this gap is an interesting open question.
\section{Interactive Local Differential Privacy}
\label{sec:interactive-ldp}

We next consider the interactive LDP model (Definition~\ref{def:interactive-local-dp}).
Unlike a non-interactive protocol, an interactive protocol can decide what to
query next after seeing earlier private reports. This fits naturally with a few voting rules--including Plurality with Runoff, STV, and Condorcet--that already proceed in ``rounds''. Our interactive protocols essentially mirror the description of these rounds, where we estimate only
the statistics needed at the current round. This results in reduced noise (and decisive margins) compared to non-interactive LDP algorithms from the previous section. %, as we discussed in more detail below. % For runoff, it first selects two
%finalists and then estimates only their head-to-head result.  For STV, it
%estimates scores only for the candidate sets encountered during the actual
%elimination process.  For Condorcet voting, it runs a knockout tournament and
%then verifies the surviving candidate.

\subsection{Plurality with Runoff}

\paragraph{Upper bound.}
The protocol first privately estimates all plurality scores using half of the
privacy budget and selects the two noisy finalists.  It then spends the
remaining budget estimating their head-to-head scores and returns the noisy
runoff winner.  Thus, the second query depends only on the privatized first-stage
output.

\begin{theorem}
\label{thm:interactive-runoff-upper}
For every $\eps>0$ and $\gamma\in(0,1/2)$, there is an interactive
$\eps$-LDP algorithm for Plurality with Runoff that is
$(\tau,\gamma)$-useful with
%\[
%    \tau = 3+2C_{\mathrm{DJW}}\kappa_{\eps/2}
%    \sqrt{2n\log\!\left(\frac{4m}{\gamma}\right)}.
%\]
$\tau = O\Paren{\max\left\{\frac{1}{\eps}, 1\right\} \cdot \sqrt{n\log(m/\gamma)}}$
\end{theorem}

\paragraph{Proof sketch.}
With probability $1-\gamma$, every score used in the two stages is
within
$\lambda=O(\kappa_{\eps/2}
\sqrt{n\log(m/\gamma)})$
of its true value.  If the protocol errs, changing at most $2\lambda+2$
rankings either removes the true winner from the runoff or makes the noisy
opponent reach the runoff and defeat it.  %The extra unit in the theorem handles
%the integer margin boundary.
\hfill $\square$

\subsection{STV}

\paragraph{Upper bound.}
STV can be naturally implemented:
In each of the $m-1$ elimination rounds, the protocol privately estimates the
plurality scores restricted to the remaining candidates and eliminates the
candidate with the smallest noisy score.
Each round uses privacy budget $\eps/(m-1)$.

The analysis uses the following structural fact.  Consider any proposed STV
elimination sequence that ends at a candidate other than the true winner. Then the margin of
victory is at most the \emph{sum}, over all rounds, of the proposed candidate's
plurality-score excess above the round minimum plus a quadratic term in the number of candidates.

This gives the following result, where an $O(m)$ factor is from the DP composition theorem and the other is due to the above fact--which accumulates the errors across rounds.

\begin{theorem}
\label{thm:interactive-stv-upper}
For every $\eps>0$ and $\gamma\in(0,1/2)$, there is an interactive
$\eps$-LDP algorithm for STV that is $(\tau,\gamma)$-useful with $\tau = O\Paren{m^2 + m \cdot \max\left\{\frac{m}{\eps}, 1\right\} \cdot \sqrt{n\log(m/\gamma)}}$.
%\[
%    \tau=4(m-1)C_{\mathrm{DJW}}\kappa_{\eps/(m-1)}
%    \sqrt{2n\log\!\left(\frac{2m^2}{\gamma}\right)}+1.
%\]
\end{theorem}

\paragraph{Proof sketch.}
A union bound over all scores released along the realized path gives
simultaneous error 
$\lambda=O\Paren{\kappa_{\eps/(m-1)}
\sqrt{n\log(m/\gamma)}}$.
The eliminated candidate is within $2\lambda$ of the true round minimum.
The tie-breaking choice therefore makes the preceding structural fact applicable, so an
erroneous output can occur only when
$\mov(\bpi,r)\le2(m-1)\lambda + O(m^2)$.
\hfill $\square$

\subsection{Condorcet}

\paragraph{Upper bound.} Our Condorcet protocol organizes the candidates into a knockout tournament.  In each
round, the remaining candidates are paired, and the protocol privately evaluates %a disjoint collection of 
their head-to-head contests. %in
%one batched release.  
After $\lceil\log_2 m\rceil$ rounds, only one candidate remains. The protocol then uses one final round to verify that the surviving candidate defeats every opponent;
otherwise it returns $\bot$.

The analysis relates pairwise margins to the $\mov$.  If $w$ is a
Condorcet winner, the margin of victory is at most one plus $w$'s smallest
head-to-head victory margin.  If the outcome is $\bot$, then, for any candidate
$c$, the margin of victory is at most $1+3\log_2 m$ plus three times the
largest head-to-head margin by which another candidate defeats or ties $c$.

\begin{theorem}
\label{thm:interactive-condorcet-upper}
For every $\eps>0$ and $\gamma\in(0,1/2)$, there is an interactive
$\eps$-LDP algorithm for Condorcet voting that is
$(\tau,\gamma)$-useful with
%\begin{align*}
%    Q&=\lceil\log_2 m\rceil+1, \\
%    \lambda_{\mathrm{C}}
%    &=C_{\mathrm{DJW}}\kappa_{\eps/Q}
%    \sqrt{2mn\log\!\left(\frac{4m^2}{\gamma}\right)}, \\
%    \tau&=2+3\log_2 m+6\lambda_{\mathrm{C}}.
%\end{align*}
%In particular, for $0<\eps\le1$,
%\[
%    \tau=O\!\left(
%    \log m+
%    \frac{\sqrt{mn\log(m/\gamma)}\log m}{\eps}
%    \right).
%\]
$\tau = \tilde{O}\Paren{\max\left\{\frac{1}{\eps}, 1\right\} \cdot \sqrt{mn \log(m/\gamma)}}$
\end{theorem}

\paragraph{Proof sketch.}
Each voter's contribution in every batch has $\ell_2$-norm at most $\sqrt m$.  With probability
$1-\gamma$, all estimated pairwise margins are therefore within $\lambda = \tilde{O}\Paren{\max\left\{\frac{1}{\eps}, 1\right\} \cdot \sqrt{mn \log(m/\gamma)}}$
$2\lambda_{\mathrm{C}}$ of their true values.  The preceding structural bounds imply that the protocol only errs if $\mov$ is at most $O(\lambda \cdot \log m)$. \hfill $\square$
%show that, when the true outcome is a candidate, an error requires margin at
%most $1+2\lambda_{\mathrm{C}}$; when the true outcome is $\bot$, an error
%requires margin at most $1+3\log_2 m+6\lambda_{\mathrm{C}}$.  
%The theorem again adds one for the integer boundary.

Note that it is crucial to use an $O(\log m)$-round knockout tournament in our implementation. If we instead naively implement the protocol in $m$ rounds (eliminating one candidate in each round), then we need to apply the DP composition theorem, resulting in a privacy budget of only $\eps/m$ per round and an error of $\tilde{O}_{\eps}(m\sqrt{n})$ instead. This would not even improve upon our non-interactive guarantee (\Cref{thm:non-int-ldp-algo-pairwise})!

\subsection{Lower Bound}

The interactive guarantees cannot improve the dependence on the number of
voters below the usual locally private mean-testing rate, even with only two
candidates.

\begin{theorem}
\label{thm:interactive-plurality-lower}
There are constants $C,\varepsilon_0>0$ such that the following
holds for two-candidate Plurality.  For
$0<\varepsilon\le\varepsilon_0$ and $0<\gamma\le1/6$, no interactive
$\varepsilon$-LDP algorithm is $\left(C\min\left\{n,\frac{\sqrt n}{\varepsilon}\right\},\gamma\right)$-useful.
\end{theorem}

\paragraph{Proof sketch.}
We reduce the fully interactive locally private simple hypothesis testing problem to the winner
publication problem.  When $\varepsilon$ is small, use the two point-mass
distributions, which produce unanimous profiles with margin $\Omega(n)$.
Otherwise, use two binary distributions separated in total variation by
$\Theta(1/(\varepsilon\sqrt n))$; concentration gives margin
$\Omega(\sqrt n/\varepsilon)$.  The fully interactive testing lower bound
of~\citet{JosephMNR19} rules out a constant-error test in both regimes.

With two candidates, Plurality, Condorcet, Maximin, Plurality with Runoff, and
STV all select the majority candidate.  Hence
Theorem~\ref{thm:interactive-plurality-lower} gives the same
$\Omega(\min\{n,\sqrt n/\eps\})$ lower bound for every voting rule studied in
this work.

\section{Conclusion}
\label{sec:conclusion}

We gave upper and lower bounds on the decisive margin of victory for private
voting mechanisms.  In the central model,
all the rules considered have the same logarithmic margin dependence, although
the optimal STV guarantee cannot be achieved in polynomial time unless
NP $\subseteq$ BPP.

In the local model, the required margin grows with the amount of information
each voter must report.  This particularly affects multi-stage rules.
Interaction reduces the cost by publishing statistics only for the realized
stages, whereas a non-interactive protocol must publish all statistics that
might be needed later.  Consequently, non-interactive protocols can require
substantially larger margins, especially for STV.

In terms of open questions, several gaps between upper and lower bounds remain in the local model. Specifically, in the interactive local model, our upper bounds for STV and Condorcet still depend polynomially on $m$. On the other hand, proving a matching lower bound seems beyond currently known techniques, since our output is simply a candidate among $m$ choices. Meanwhile, existing work (e.g. \cite{AcharyaCST23}) requires $2^{m^{\Omega(1)}}$ different possible outcomes to arrive at a lower bound that is polynomial in $m$. Thus, closing this gap might have implications beyond private voting.

\subsection*{Acknowledgment} We thank Jittat Fakcharoenphol for helpful discussion during the early stages of this work.

\newpage

\bibliography{ref}

\newpage

\appendix

\section{Margin-of-Victory Certificates}
\label{app:structural-mov}

This section collects the voting-rule facts used in the later utility proofs.
For each rule, we show that an inaccurate winner can be reported only when the
margin of victory is small relative to the errors in the estimated scores or
pairwise margins.  The later analyses therefore have two steps: they bound the
estimation error of the private algorithm and then apply the corresponding
result from this section.  The same second step applies to both the
non-interactive and interactive algorithms.

Throughout this appendix, when candidate scores are estimated,
$\widetilde{s}(c)$ denotes the estimated score of candidate $c$, and
$\widetilde{s}=(\widetilde{s}(c))_{c\in\cC}$ denotes the vector of these
estimates.  Each result specifies the corresponding true score being
estimated.

\subsection{Plurality}

The margin of victory for positional scoring rules, including Plurality, can
be computed exactly~\citep[Theorem~5]{Xia12}.  The following elementary
certificate is also closely related to the score-gap bound used by
\citet[Lemma~3]{BHATTACHARYYA2021103476}.

\begin{lemma}[Plurality certificate]
\label{lem:margin-plurality}
Suppose $w$ is the Plurality winner on $\bpi$.  For every $c\ne w$,
\[
    \mov(\bpi,r)
    \le
    \left\lfloor
    \frac{\scoreplu(\bpi,w)-\scoreplu(\bpi,c)}{2}
    \right\rfloor+1.
\]
Consequently, if the estimated Plurality scores satisfy
$|\widetilde{s}(c)-\scoreplu(\bpi,c)|\le\eta$ for every candidate $c$, then a
candidate maximizing $\widetilde{s}$ can differ from $w$ only if
\[
    \mov(\bpi,r)\le\eta+1.
\]
\end{lemma}

\begin{proof}
Let
$\Delta=\scoreplu(\bpi,w)-\scoreplu(\bpi,c)$ and
$t=\lfloor\Delta/2\rfloor+1$.  Change $t$ rankings that place $w$ first so
that they place $c$ first.  This decreases the score gap between $w$ and $c$
by $2t>\Delta$, so $c$ strictly outranks $w$ in the resulting profile.

For the second claim, if $c$ maximizes the estimated score, then
\[
    \scoreplu(\bpi,w)-\scoreplu(\bpi,c)\le2\eta.
\]
The first claim now gives $\mov(\bpi,r)\le\eta+1$.
\end{proof}

\subsection{Plurality with Runoff}

Exact computation of the margin of victory for Plurality with Runoff was
studied by \citet[Theorem~7]{Xia12}.  A related certificate in terms of
first-round scores and pairwise margins appears in
\citet[Lemma~8]{BHATTACHARYYA2021103476}; the formulation below is tailored to
separate errors in the two stages of the noisy computation.

\begin{lemma}[Runoff certificate]
\label{lem:margin-runoff}
Suppose a noisy computation for Plurality with Runoff uses first-round estimates
$\widetilde{s}(c)$ satisfying
\[
    |\widetilde{s}(c)-\scoreplu(\bpi,c)|\le\eta
    \qquad\text{for every }c\in\cC,
\]
selects the two candidates with largest estimated scores, and uses an estimate
$\widehat D(a,b)$ of their pairwise margin satisfying
\[
    |\widehat D(a,b)-D_{\bpi}(a,b)|\le\zeta.
\]
If the noisy runoff winner differs from the true runoff winner, then
\[
    \mov(\bpi,r)
    \le
    2+\max\left\{2\eta,\eta+\frac{\zeta}{2}\right\}.
\]
\end{lemma}

\begin{proof}
Let $w$ be the true runoff winner, let $d$ be the other true finalist, and
let $a,b$ be the noisy finalists.

First suppose $w\notin\{a,b\}$.  For each $x\in\{a,b\}$,
\[
    \scoreplu(\bpi,w)-\scoreplu(\bpi,x)\le2\eta.
\]
Write $q=\scoreplu(\bpi,w)$ and set
\[
    t_x=\max\left\{0,
    \left\lfloor
    \frac{\scoreplu(\bpi,w)-\scoreplu(\bpi,x)}{2}
    \right\rfloor+1\right\}.
\]
Then $t_x\le\eta+1$, so $t_a+t_b\le2\eta+2$.  If
$q\ge t_a+t_b$, change disjoint sets of $t_a$ and $t_b$ rankings that place
$w$ first so that they instead place $a$ and $b$ first, respectively.  Each
change decreases $w$'s score by one and increases the chosen candidate's score
by one.  Thus both $a$ and $b$ strictly outrank $w$, so $w$ misses the runoff.

It remains to consider $q<t_a+t_b$.  Since $w$ is the true runoff winner, we
must have $q>0$: if $q=0$, at most one candidate can have positive first-round
score while $w$ is a finalist, and that candidate is ranked first by every
voter and therefore defeats $w$ head to head.  Change all $q$ rankings that
place $w$ first so that they place $a$ first and $w$ last.  Candidate $w$ now
has score zero and $a$ has positive score.  If $w$ misses the runoff, the
outcome has changed.  Otherwise, no candidate other than $a$ has positive
score, since any such candidate would join $a$ in the runoff ahead of $w$.
Hence every voter ranks $a$ first, and $a$ defeats $w$ unanimously.  This
construction uses
\[
    q<t_a+t_b\le2\eta+2
\]
changes.

Now suppose $w$ is a noisy finalist but loses the noisy runoff to a candidate
$a$.  If $a\ne d$, then
\[
    \scoreplu(\bpi,d)-\scoreplu(\bpi,a)\le2\eta.
\]
Set
\[
    t=\max\left\{0,
    \left\lfloor
    \frac{\scoreplu(\bpi,d)-\scoreplu(\bpi,a)}{2}
    \right\rfloor+1\right\}.
\]
Then $t\le\eta+1$.  If $\scoreplu(\bpi,d)>0$, there are at least $t$
rankings with $d$ first.  Change $t$ of them to place $a$ first and place $a$
above $w$.  Candidate $a$ then strictly outranks $d$ and is therefore a true
finalist.  The only exceptional case is
$\scoreplu(\bpi,d)=\scoreplu(\bpi,a)=0$, when $t=1$; changing one ranking
with $w$ first in the same way makes $a$ a finalist.  Such a ranking exists
because, as argued above, a true runoff winner has positive first-round score.
If $a=d$, take $t=0$, since $a$ is already a true finalist.

If these changes remove $w$ from the runoff, the outcome has already changed.
Otherwise, the true finalists are now $w$ and $a$.  Because $a$ wins their
noisy comparison and its estimated pairwise margin has error at most $\zeta$,
\[
    D_{\bpi}(w,a)\le\zeta.
\]
The first-round changes all place $a$ above $w$, so they cannot increase this
margin.  Let $D'(w,a)$ be the margin after those changes and set
\[
    u=\max\left\{0,
    \left\lfloor\frac{D'(w,a)}{2}\right\rfloor+1\right\}.
\]
Then $u\le\zeta/2+1$.  Change $u$ additional rankings that place $w$ above
$a$ by moving $a$ immediately above $w$.  This cannot decrease $a$'s
first-round score and makes $D'(w,a)-2u<0$, so $a$ wins the runoff.  The total
number of changes in this case is at most
\[
    t+u\le\eta+\frac{\zeta}{2}+2.
\]
Taking the larger of the two case bounds proves the claim.
\end{proof}

\subsection{Condorcet}

\begin{lemma}[Condorcet certificate]
\label{lem:margin-condorcet}
Suppose $w\in\cC$ is the Condorcet winner on profile $\bpi$.  Then
\[
    \mov(\bpi,r)
    \le
    1+\min_{c\in\cC\setminus\{w\}}D_{\bpi}(w,c).
\]
Alternatively, suppose the Condorcet rule outputs $\bot$ on $\bpi$.  Then,
for every candidate $c\in\cC$,
\[
    \mov(\bpi,r)
    \le
    1+3\log_2m
    +3\max_{a\in\cC\setminus\{c\}}D_{\bpi}(a,c).
\]
\end{lemma}

\begin{proof}
First consider the case where the Condorcet rule outputs a candidate $w$.  Let
\begin{align*}
    \Delta
    &\coloneq \min_{c\in\cC\setminus\{w\}}D_{\bpi}(w,c), \\
    c^\star
    &\coloneq \arg\min_{c\in\cC\setminus\{w\}}D_{\bpi}(w,c).
\end{align*}
Changing the relative order of $w$ and $c^\star$ in
$\lfloor\Delta/2\rfloor+1$ voters who place $w$ above $c^\star$ makes
$D_{\bpi'}(w,c^\star)<0$ in the resulting profile $\bpi'$.  Thus $w$ is no
longer a Condorcet winner, and
\[
    \mov(\bpi,r)
    \le\lfloor\Delta/2\rfloor+1
    \le1+\Delta.
\]

Now suppose the rule outputs $\bot$.  Fix a candidate $c\in\cC$ and define
\[
    \Delta_c
    \coloneq
    \max_{a\in\cC\setminus\{c\}}D_{\bpi}(a,c).
\]
For each $a\in\cC\setminus\{c\}$, let $V_a\subseteq[n]$ be the set of voters
who rank $a$ above $c$, and write
$\overline{V_a}=[n]\setminus V_a$.  By definition,
$|V_a|\le(n+\Delta_c)/2$.  We will find a set $X$ of at most
$1+3\log_2m+3\Delta_c$ voters such that, after moving $c$ to the top of those
voters' rankings, candidate $c$ defeats every $a\ne c$ head to head.

Let
\[
    f(X)
    \coloneq
    \sum_{a\in\cC\setminus\{c\}}(2^{r_a}-1),
\]
where
\[
    r_a
    \coloneq
    \max\left\{
        0,
        1+\left\lfloor\frac{n}{2}\right\rfloor
        -|\overline{V_a}\cup X|
    \right\}.
\]
We claim that whenever $f(X)>0$, there exists $x\in[n]\setminus X$ such that
\[
    f(X\cup\{x\})\le\frac34f(X).
\]
If $r_a>0$, then
$|\overline{V_a}\cup X|\le\lfloor n/2\rfloor$, so at least half of the
remaining voters $x\in[n]\setminus X$ rank $a$ above $c$.  Adding any such
voter to $X$ decreases $a$'s contribution from $2^{r_a}-1$ to
$2^{r_a-1}-1$.  Therefore, for a uniformly random $x\in[n]\setminus X$,
\begin{align*}
    f(X)-\mathbb{E}_x[f(X\cup\{x\})]
    &\ge\frac12
        \sum_{\substack{a\in\cC\setminus\{c\}\\r_a>0}}
        2^{r_a-1} \\
    &\ge\frac14f(X).
\end{align*}
Hence $\mathbb{E}_x[f(X\cup\{x\})]\le\frac34f(X)$, so some
$x\in[n]\setminus X$ satisfies the claimed inequality.

At $X=\emptyset$, parity of the pairwise margins gives
$r_a\le\max\{1,\Delta_c\}$, and hence
$f(\emptyset)\le m2^{\Delta_c}$.  Since $f$ is integer-valued, the process
stops with $f(X)=0$ once $f(X)<1$.  A greedy algorithm that repeatedly chooses
the best voter returns a set $X$ of size at most
\[
    1+\log_{4/3}(m2^{\Delta_c})
    \le1+3\log_2m+3\Delta_c.
\]
Let $\bpi'$ be obtained from $\bpi$ by moving $c$ to the top of every
ranking indexed by $X$.  Since $f(X)=0$, we have $r_a=0$ for every
$a\ne c$.  By the definition of $r_a$, this implies
\[
    |\overline{V_a}\cup X|
    \ge 1+\left\lfloor\frac n2\right\rfloor.
\]
Every voter in $\overline{V_a}$ already ranks $c$ above $a$, and every voter
in $X$ does so after the modification.  Hence, in $\bpi'$, strictly more than
half of the voters rank $c$ above $a$.  This holds for every $a\ne c$, so
$c$ is the Condorcet winner on $\bpi'$.  The original outcome was $\bot$;
therefore changing at most $|X|$ rankings changes the outcome, and
\[
    \mov(\bpi,r)
    \le |X|
    \le1+3\log_2m+3\Delta_c,
\]
as claimed.
\end{proof}

\subsection{Maximin}

The relationship between the Maximin score gap and the margin of victory is
due to \citet[Theorem~12]{Xia12}; see also the related score-gap certificate
of \citet[Lemma~5]{BHATTACHARYYA2021103476}.  We record the resulting
consequence for noisy pairwise margins.

\begin{lemma}[Maximin certificate]
\label{lem:margin-maximin}
Let $w$ and $w'$ have the largest and second-largest maximin scores,
respectively, and define
\[
    \Delta_{\mathrm{mm}}
    =\scmaximin(w)-\scmaximin(w').
\]
Then
\[
    \frac{\mov(\bpi,r)-1}{2}
    \le\frac{\Delta_{\mathrm{mm}}}{4}
    \le\mov(\bpi,r).
\]
In particular,
\[
    \mov(\bpi,r)\le\frac{\Delta_{\mathrm{mm}}}{2}+1.
\]
Moreover, if all pairwise margins are estimated to error at most $\eta$, a
candidate maximizing the noisy maximin score can differ from $w$ only if
\[
    \mov(\bpi,r)\le\eta+1.
\]
\end{lemma}

\begin{proof}
The first two inequalities are the two-approximation guarantee of
\citet[Theorem~12]{Xia12}, with the additive one accounting for the fixed
tie-breaking rule.  Rearranging the first inequality gives the displayed upper
bound on $\mov$.

Suppose every pairwise margin has error at most $\eta$, and let
$\tscmaximin(a)$ denote the resulting noisy maximin score of candidate $a$.
Every pairwise margin entering the minimum that defines a candidate's
maximin score changes by at most $\eta$.  Therefore that minimum also changes
by at most $\eta$.  Thus, for every $a\in\cC$,
\[
    |\tscmaximin(a)-\scmaximin(a)|\le\eta.
\]
Let $\widehat w\ne w$ maximize the noisy maximin score.  The optimality of
$\widehat w$ for the noisy scores gives
\begin{align*}
    \scmaximin(w)
    &\le \tscmaximin(w)+\eta \\
    &\le \tscmaximin(\widehat w)+\eta \\
    &\le \scmaximin(\widehat w)+2\eta.
\end{align*}
Because $w'$ has the second-largest true maximin score,
$\scmaximin(w')\ge\scmaximin(\widehat w)$.  Therefore
\[
    \Delta_{\mathrm{mm}}
    =\scmaximin(w)-\scmaximin(w')
    \le2\eta.
\]
The preceding bound now gives
$\mov(\bpi,r)\le\Delta_{\mathrm{mm}}/2+1\le\eta+1$.
\end{proof}

\subsection{Single Transferable Vote}

For STV, we convert a near-minimum elimination path into an upper bound on the
margin of victory and then derive the corresponding error certificate.  A
closely related elimination-path bound appears in
\citet[Lemma~9]{BHATTACHARYYA2021103476}.  Our formulation includes an
additive $m^2$ term to enforce the desired path under fixed tie-breaking.

\begin{lemma}[STV upper certificate]
\label{lem:margin-stv}
Suppose $w$ is the STV winner on $\bpi$, and consider distinct candidates
$c_1,\ldots,c_m$ with $c_m\ne w$.  Define
\[
    \cC_{-i}=\cC\setminus\{c_1,\ldots,c_i\},
    \qquad
    s_i^\star=\min_{c\in\cC_{-i}}
    \scoreplu(\bpi|_{\cC_{-i}},c).
\]
Then
\[
    \mov(\bpi,r)
    \le m^2+\sum_{i=0}^{m-1}
    \left(
    \scoreplu(\bpi|_{\cC_{-i}},c_{i+1})-s_i^\star
    \right).
\]
\end{lemma}

\begin{proof}
For $i=0,\ldots,m-1$, let
\[
    \Delta_i
    \coloneq
    \scoreplu(\bpi|_{\cC_{-i}},c_{i+1})-s_i^\star
    \ge0.
\]
We will show that by changing
$m^2+\sum_{i=0}^{m-1}\Delta_i$ rankings, we can make $c_m$ win.

First observe that the sequence $(s_i^\star)$ is nondecreasing, since
remaining candidates can only gain votes from one round to the next.

We construct by induction profiles
$\bpi^{(0)},\bpi^{(1)},\ldots,\bpi^{(m)}$, with
$\bpi^{(0)}=\bpi$.  For each $i\in\{0,\ldots,m-1\}$, we modify exactly
$\Delta_i$ rankings in $\bpi^{(i)}$ to construct $\bpi^{(i+1)}$.  We
maintain the following properties:
\begin{itemize}
    \item for each $j\in\{0,\ldots,i-1\}$, candidate $c_{j+1}$ has
    plurality score exactly $s_j^\star$ in
    $\bpi^{(i)}|_{\cC_{-j}}$;
    \item for each $j\in\{0,\ldots,m-1\}$,
    \[
        \min_{c\in\cC_{-j}}
        \scoreplu(\bpi^{(i)}|_{\cC_{-j}},c)=s_j^\star;
    \]
    \item for each $j\in\{i,\ldots,m-1\}$ and every
    $c\in\cC_{-j}$,
    \[
        \scoreplu(\bpi^{(i)}|_{\cC_{-j}},c)
        =\scoreplu(\bpi|_{\cC_{-j}},c).
    \]
\end{itemize}

The base case $i=0$ follows from $\bpi^{(0)}=\bpi$.  Fix
$i\in\{0,\ldots,m-1\}$ and assume the properties hold for $\bpi^{(i)}$.
Consider the rankings whose top candidate in $\cC_{-i}$ is $c_{i+1}$.
Choose $\Delta_i$ of them that are last to give their vote to $c_{i+1}$,
meaning that the first round in which they contribute to $c_{i+1}$ is as late
as possible.  Such a ranking has the form
\[
    \cdots\succ c_{i+1}\succ\cdots\succ d\succ\cdots,
\]
where $d$ is the second preferred candidate below $c_{i+1}$ in $\cC_{-i}$.
Swap $c_{i+1}$ and $d$, leaving the relative order of every other candidate
unchanged.

Each change decreases the score of $c_{i+1}$ in $\cC_{-i}$ by one and
transfers that vote to $d$.  Hence
\[
    \scoreplu(\bpi^{(i+1)}|_{\cC_{-i}},c_{i+1})=s_i^\star.
\]
For $j>i$, the restrictions to $\cC_{-j}$ are unchanged, because once
$c_{i+1}$ is removed the relative order of the remaining candidates is the
same.

It remains to verify the earlier rounds $j<i$.  Such a modification can only
lower the score of $c_{i+1}$ and raise the score of $d$.  Fix $j<i$ and let
\[
    u_j=
    \scoreplu(\bpi^{(i)}|_{\cC_{-j}},c_{i+1}).
\]
The rankings contributing to $c_{i+1}$ in round $j$ are precisely the
earlier-contributing rankings among those that contribute to it in round $i$.
Because we modify the $\Delta_i$ latest-contributing rankings, at least
$\min\{u_j,s_i^\star\}$ of its round-$j$ votes remain.  The induction
hypothesis gives $u_j\ge s_j^\star$, while the monotonicity of the round
minima gives $s_i^\star\ge s_j^\star$.  Hence the new round-$j$ score of
$c_{i+1}$ is at least $s_j^\star$.  No other score decreases, so no score
falls below the original round minimum, and the induction properties hold for
$\bpi^{(i+1)}$.

Consequently, in $\bpi^{(m)}$, candidate $c_{i+1}$ is among the minimum-score
candidates in round $i+1$ for every $i$.  The number of changed rankings so
far is at most $\sum_i\Delta_i$.

The intended candidate may still be tied for the minimum.  To make the
elimination order independent of tie-breaking, we add a staircase.  Let
$J=\{j\in[m]:c_j\ne w\}$ and
\[
    L=\sum_{j\in J}j\le\frac{m(m+1)}{2}\le m^2.
\]
Suppose first that at least $L$ rankings in $\bpi^{(m)}$ place $w$ first.
Choose pairwise disjoint sets $X_j$ of such rankings, with $|X_j|=j$ for
every $j\in J$.  In each ranking in $X_j$, move $c_j$ immediately above $w$,
preserving the relative order of all other candidates.

Before $w$ is eliminated, candidate $c_j\ne w$ receives an additional $j$
votes in every round in which it remains.  Thus, in round $i+1$, the
staircase bonuses of the surviving candidates
$c_{i+1},\ldots,c_m$ other than $w$ are
\[
    i+1,i+2,\ldots,m.
\]
Since every surviving candidate previously had score at least $s_i^\star$,
while $c_{i+1}$ had score exactly $s_i^\star$, either $w$ or $c_{i+1}$ is
eliminated in that round.  Therefore the elimination order follows
$c_1,c_2,\ldots,c_{m-1}$ until $w$ is eliminated.  Since $c_m\ne w$, the
resulting winner differs from $w$.

It remains to handle the case in which fewer than $L$ rankings in
$\bpi^{(m)}$ place $w$ first.  Change all of these rankings by moving $c_m$
to the top and $w$ to the bottom.  Candidate $w$ then has zero score in the
first STV round.  A zero-score candidate cannot win STV: while another
zero-score candidate is eliminated, no votes transfer, and once $w$ is the
only zero-score candidate, it is the unique minimum and is eliminated.  Thus
these fewer than $L\le m^2$ additional changes also produce a winner other
than $w$.

The total number of changes is at most
\[
    \frac{m(m+1)}{2}+\sum_{i=0}^{m-1}\Delta_i
    \le m^2+\sum_{i=0}^{m-1}\Delta_i,
\]
proving the claim.
\end{proof}

\begin{corollary}[STV error certificate]
\label{cor:margin-stv-error}
Suppose an approximate STV computation starts with $S_0=\cC$.  In each round
$i\in\{0,\ldots,m-2\}$, it computes estimates
$(\widetilde{s}_i(c))_{c\in S_i}$ of the round scores, eliminates a candidate
\[
    c_{i+1}\in\arg\min_{c\in S_i}\widetilde{s}_i(c),
    \qquad
    S_{i+1}=S_i\setminus\{c_{i+1}\},
\]
and finally returns the unique candidate $c_m\in S_{m-1}$.  Suppose
$c_m$ differs from the true STV winner and, for every realized round $i$ and
every active candidate $c\in S_i$,
\[
    |\widetilde{s}_i(c)-\scoreplu(\bpi|_{S_i},c)|\le\eta.
\]
Then
\[
    \mov(\bpi,r)\le m^2+2(m-1)\eta.
\]
\end{corollary}

\begin{proof}
Let $w$ be the true STV winner.  By assumption, $c_m\ne w$.  For
$i\in\{0,\ldots,m-2\}$, the active set before eliminating $c_{i+1}$ is
\[
    S_i=\cC_{-i}=\cC\setminus\{c_1,\ldots,c_i\}.
\]
Choose a candidate $d_i\in\cC_{-i}$ having minimum true score in this round,
so $\scoreplu(\bpi|_{\cC_{-i}},d_i)=s_i^\star$.  Since the approximate
computation eliminates a candidate of minimum estimated score,
\begin{align*}
    \scoreplu(\bpi|_{\cC_{-i}},c_{i+1})
    &\le \widetilde{s}_i(c_{i+1})+\eta \\
    &\le \widetilde{s}_i(d_i)+\eta \\
    &\le \scoreplu(\bpi|_{\cC_{-i}},d_i)+2\eta \\
    &=s_i^\star+2\eta.
\end{align*}
Thus the score excess appearing in Lemma~\ref{lem:margin-stv} is at most
$2\eta$ in each of the $m-1$ elimination rounds.  When $i=m-1$, only $c_m$
remains, so its score equals the minimum and the final excess is zero.
Summing these bounds and applying Lemma~\ref{lem:margin-stv} gives
\[
    \mov(\bpi,r)
    \le m^2+2(m-1)\eta.
\]
\end{proof}

\subsubsection{Necessity of the quadratic term.}

Lemma~\ref{lem:margin-stv} bounds the STV margin of victory by $m^2$ plus, for
each round, the difference between the Plurality score of the candidate
proposed for elimination and the minimum Plurality score among the remaining
candidates.  The
following remark shows that the $m^2$ term cannot be removed in general: even
when every such difference is zero, changing the STV winner may
require $\Omega(m^2)$ ranking changes.

\begin{remark}
\label{rem:stv-quadratic}
For every $m\ge4$, there is a profile $\bpi$ on $m$ candidates
with STV winner $w$ and distinct candidates $c_1,\ldots,c_m$, where
$c_m\ne w$.  For $i\in\{0,\ldots,m-1\}$, define
\[
    \cC_{-i}=\cC\setminus\{c_1,\ldots,c_i\},
    \qquad
    s_i^\star=\min_{c\in\cC_{-i}}
    \scoreplu(\bpi|_{\cC_{-i}},c).
\]
Then, for every $i\in\{0,\ldots,m-1\}$,
\[
    \scoreplu(\bpi|_{\cC_{-i}},c_{i+1})-s_i^\star = 0,
\]
but
\[
    \mov(\bpi,r)\ge\frac{(m-3)^2}{64}.
\]
\end{remark}

\begin{proof}
Let $k=m-2$, let
$\cC=\{a,b,x_1,\ldots,x_k\}$, and fix the tie-breaking order
\[
    a\prec x_1\prec\cdots\prec x_k\prec b,
\]
where $y\prec z$ means that $y$ is eliminated before $z$ when their scores
are equal.  Set $R=k^2$, and consider the profile consisting of
\begin{itemize}
    \item $R$ rankings of the form
    $a\succ b\succ x_1\succ\cdots\succ x_k$;
    \item $(m-1)R$ rankings of the form
    $b\succ a\succ x_1\succ\cdots\succ x_k$;
    \item for each $i\in[k]$, $R$ rankings of the form
    \[
        x_i\succ x_k\succ\cdots\succ x_{i+1}
        \succ a\succ b\succ x_{i-1}\succ\cdots\succ x_1.
    \]
\end{itemize}

Candidate $b$ wins: candidate $a$ is eliminated in the first round by the
tie-breaking rule, after which $b$ has strictly more than half of the votes.
On the other hand, along the proposed order
$x_k,\ldots,x_1,b,a$, the next candidate is tied for the minimum in each
restricted election.  Thus the displayed score excess is zero.

It remains to show that the margin of victory of $b$ is at least
$(m-3)^2/64$.  Set
\[
    \ell_0=\frac{(k-1)^2}{64}=\frac{(m-3)^2}{64},
\]
and let $\bpi'$ differ from $\bpi$ on $T<\ell_0$ rankings.  Since
$T<R/4$, candidate $b$ cannot be eliminated while at least two other
candidates remain: its score is at least $(m-1)R-T$, whereas the total score
of all other candidates is at most $(m-1)R+T$, so one of them
has score at most
\[
    \frac{(m-1)R+T}{2}<(m-1)R-T.
\]

If $b$ faces some $x_i$ in the final round, then after $a$ is eliminated,
$b$ has at least $mR-T$ votes, while $x_i$ has at most
$kR+T<mR-T$ votes.  Thus $b$ still wins.  It remains to consider a final round
between $a$ and $b$.  This requires the elimination order to begin with
$x_k,x_{k-1},\ldots,x_q$, where $q=\lceil k/2\rceil+1$.  Otherwise, suppose
the order first breaks when $x_i$ is eliminated while $x_j$ is the
highest-index remaining candidate with $j>i$ and $j\ge q$.  The votes of
$x_{j+1},\ldots,x_k$ have transferred to $a$, while eliminating $x_i$
transfers its votes to $x_j$, which then has at least $2R-T>R+T$ votes.  All
other remaining $x$-candidates are eliminated before $x_j$ and transfer to
it.  When only $a,b,x_j$ remain, candidate $a$ has at most
\[
    (k-j)R+T\le(\lfloor k/2\rfloor-1)R+T
\]
votes, while $x_j$ has at least $jR-T\ge qR-T$, so $a$ cannot reach the
final round.

We now show that this required prefix itself needs at least $\ell_0$ changes.
Write
\[
    y_1,\ldots,y_p=x_k,\ldots,x_q,
    \qquad p=\lfloor k/2\rfloor.
\]
If $p=1$, the original first-round scores of $a$ and $y_1=x_k$ are both $R$,
and the tie-breaking rule eliminates $a$ before $y_1$.  Reversing this order
requires at least one change, which is more than $\ell_0$.  Hence suppose
$p\ge2$.

For each $i\in[p]$, let $a_i$ be the score of $y_i$ when it is eliminated,
and let $B_i$ be the block of $R$ original rankings that place $y_i$ first.
The sequence $(a_i)$ is strictly increasing.  Indeed, when $y_i$ is
eliminated, $y_{i+1}$ remains and precedes $y_i$ in the tie-breaking order, so
$y_i$ must have strictly smaller score; the score of $y_{i+1}$ cannot decrease
before its own elimination.  Let $r$ be the largest index with $a_r<R$, and
set $r=0$ if no such index exists.

Suppose first that $r\ge p/2$.  Since the $a_i$ are strictly increasing
integers and $a_r\le R-1$,
\[
    R-a_i\ge r-i+1
    \qquad\text{for every }i\le r.
\]
At least $R-a_i$ rankings in $B_i$ must have been changed; otherwise more
than $a_i$ unchanged rankings would still give their vote to $y_i$ when it is
eliminated.  The blocks $B_1,\ldots,B_r$ are disjoint, and hence
\[
    T\ge\sum_{i=1}^r(R-a_i)
    \ge\frac{r(r+1)}2
    \ge\frac{p^2}{8}
    \ge\frac{(k-1)^2}{32}
    \ge \ell_0.
\]

Now suppose $r<p/2$.  Put $h=p-r$ and choose the integer round
\[
    s=r+\left\lfloor\frac h2\right\rfloor.
\]
Then $r<s<p$.  Strict increase gives
$a_s\ge R+s-r-1$.  When $y_s$ is eliminated, each of the
$p-s$ remaining candidates $y_{s+1},\ldots,y_p$ must therefore have score at
least $R+s-r$.  Along this elimination prefix, an unchanged ranking outside
the candidate's own block does not contribute to any of these remaining
candidates.  Thus their total score excess over the original value $R$ must
come from changed rankings, and one changed ranking contributes at most one
such vote in this round.  Consequently,
\[
    T\ge(p-s)(s-r)
    =\left\lceil\frac h2\right\rceil
     \left\lfloor\frac h2\right\rfloor
    \ge\frac{p^2}{16}
    \ge\frac{(k-1)^2}{64}=\ell_0.
\]
Here the penultimate inequality uses $h=p-r>p/2$, and the last uses
$2p\ge k-1$.
Both cases contradict $T<\ell_0$.  Therefore $b$ remains the winner after
fewer than $\ell_0$ ranking changes, proving
$\mov(\bpi,r)\ge(m-3)^2/64$.
\end{proof}

\section{Hardness of Approximation for STV}
\label{sec:non-private-STV}

We give a non-private hardness result for STV.  This result is used twice: first
to show strong inapproximability of MoV-related optimization problems, and then
to rule out efficient centrally private STV algorithms with strong usefulness
guarantees.

For this appendix, we use a directed variant of MoV.  For each
$a \in \cC \cup \{ \bot \}$, let $\move(\bpi, r, a)$ denote the smallest
integer $k$ such that there exists a $k$-neighbor $\bpi'$ of $\bpi$ with
$r(\bpi') = a$.  The problem \bribery{}~\cite{FaliszewskiHH09} asks us to
approximately compute $\move(\bpi,r,a)$ from $(\bpi,r)$ and a desired outcome
$a\ne r(\bpi)$.  This problem is sometimes called \emph{constructive
bribery}, while the problem of computing $\mov$ is also known as \emph{destructive bribery}.  We have
$\move(\bpi, r, r(\bpi))=0$ and
\[
\mov(\bpi, r)=\min_{a \neq r(\bpi)} \move(\bpi, r, a).
\]

\begin{theorem}[Used in Theorem~\ref{thm:stv-no-efficient-dp}]
\label{thm:stv}
Let $\zeta > 0$ be an arbitrary constant.  Given an STV instance $(\bpi,r)$
with winner $w=r(\bpi)$ and a candidate $a\ne w$, it is NP-hard to distinguish
between the following two cases:
\begin{itemize}
\item (YES) For every positive integer $k \leq n^{1 - \zeta}$, there exists a $k$-neighbor $\bpi'$ of $\bpi$ such that $r(\bpi') = a$ and $\mov(\bpi', r) \geq \left\lceil \frac{k}{2} \right\rceil$. In particular, this implies $\mov(\bpi, r) \leq \move(\bpi, r, a)  \leq 1$.
\item (NO) $\move(\bpi,r,a) = \mov(\bpi, r) \geq n^{1 - \zeta}$.
\end{itemize}
Furthermore, this holds even when $m \leq n^{\zeta}$.
\end{theorem}

Before proving the theorem, we record two immediate approximation-hardness
corollaries.  The additional YES-case guarantee
$\mov(\bpi',r)\ge\lceil k/2\rceil$ is not needed for these corollaries; it is
included for the later DP lower bound.

\begin{corollary} \label{cor:inapprox-commov}
For any constant $\zeta > 0$, it is NP-hard to approximate $\mov$ for STV to within a factor of $n^{1 - \zeta}$.
\end{corollary}

\begin{corollary} \label{cor:inapprox-bribery}
For any constant $\zeta > 0$, it is NP-hard to approximate $\bribery$ for STV to within a factor of $n^{1 - \zeta}$.
\end{corollary}

Both results are nearly tight: changing all votes gives a trivial
$n$-approximation.  They also strengthen the earlier factor-$2$ hardness from
\cite{Xia12}, which builds on the NP-hardness of distinguishing
$\mov(\bpi,r)=1$ from $\mov(\bpi,r)=2$~\cite{BO91}.

The reduction is from 3SAT.  The input is a 3CNF formula $\Phi$ with $M$ clauses
over $N$ variables; each clause is a disjunction of three literals.  Let $\cX$
be the variable set.  An assignment is a function $\psi:\cX\to\{0,1\}$, and it
satisfies $\Phi$ if every clause is satisfied.

We say that $\Phi$ has \emph{bounded literal degree $d$} if every literal appears
in at most $d$ clauses.  We use the following result of
\cite{BKS03}.\footnote{See also \cite{DarmannD21}, which gives a simplified and
generalized reduction.}
%\footnote{Strictly speaking, \cite{BKS03} proves hardness for the case where each \emph{variable} appears in at most \emph{four} clauses. However, we may assume that every literal appears in at most \emph{three} clauses in their instances, since, otherwise, we can simply assign the variable according to the literal and remove all the associated clauses.}:

\begin{theorem}[\citealt{BKS03}]
It is NP-hard to decide whether a given 3CNF formula $\Phi$ with bounded literal degree 2 is satisfiable.
\end{theorem}

%We will use the following celebrated result of Hastad as a starting point of our reduction.

%\begin{theorem}[\cite{Hastad01}]
%For any $\gamma > 7/8$, it is NP-hard to, given a 3CNF formula $\Phi$, distinguish between the following cases:
%\begin{itemize}
%\item (YES) There exists a satisfying assignment of $\Phi$.
%\item (NO) Any assignment of $\Phi$ has value at most $\gamma$.
%\end{itemize}
%\end{theorem}

The reduction below is inspired by the set-cover-based reduction of
\cite{BO91}.  Reducing from 3SAT gives a smaller construction and makes the
amplification needed for approximation hardness more transparent.

\subsection{Reduction and Proof of Theorem~\ref{thm:stv}}

Consider a 3CNF formula $\Phi$ with bounded literal degree $2$.  Let
$x_1,\ldots,x_N$ be its variables and $\mathsf{c}_1,\ldots,\mathsf{c}_M$ its
clauses.  Since each clause has three literals and each literal appears at most
twice,
$\frac{N}{3} \le M \le \left\lceil \frac{4N}{3} \right\rceil.$
The candidate set $\cC$ contains the following $4N+2$ candidates:
\begin{itemize}
\item \emph{Winner Candidates:} two candidates $w$ and $a$, where $w$ is the
winner in the constructed profile and $a$ is the alternate winner.
\item \emph{Literal Candidates:} candidates $y_{x_i}$ and $y_{\ox_i}$ for each
variable $x_i$.
\item \emph{Selection Candidates:} candidates $z_{x_i}$ and $z_{\ox_i}$ for each
variable $x_i$.
%\item \emph{Clause Candidates: } We create $M$ candidates $c_1, \dots, c_M$ corresponding to the clauses of $\Phi$.
\end{itemize}
Let
$t=\left\lceil1000\cdot(2000N^2)^{1/\zeta}\right\rceil$ and $D=10N$.  The
profile $\bpi$ consists of the following voters; unspecified suffixes may be
ordered arbitrarily.
\begin{itemize}
\item \emph{Winner voters:} $4Dt$ voters with preferences $w \succ a \succ \cdots$.
\item \emph{Alternate Voters:} $4Dt-1$ voters with preferences $a \succ w \succ \cdots$.
\item For each literal $\ell$ of variable $x_i$, with negation $\oell$, create:
\begin{itemize}
\item \emph{Assignment Voters for $\ell$:} $(D+5i-4)t$ voters ranking
$z_{\ell} \succ y_{\oell} \succ y_{\ell} \succ w \succ a \succ \cdots$.
\item \emph{Selection Voters for $\ell$:} $4t$ voters ranking
$z_{\ell} \succ z_{\oell} \succ y_{\ell} \succ y_{\oell} \succ w \succ \cdots$.
\item \emph{Literal Voters for $\ell$:} $(D+5i+1)t$ voters ranking
$y_{\ell} \succ y_{\oell} \succ w \succ \cdots$.
\end{itemize}
\item \emph{Clause Voters:} for every clause $\mathsf{c}_i$ with literals
$\ell_{i,1},\ell_{i,2},\ell_{i,3}$, create $t$ voters ranking
$y_{\ell_{i,1}} \succ y_{\ell_{i,2}} \succ y_{\ell_{i,3}} \succ w \succ a \succ \cdots$.
%\item \emph{Literal Voters: } For every variable $x_i$, create $2mt$ voters of the form $x_i \succ \ox_i \succ w \succ a \succ \cdots$ and $2mt$ voters of the form  $\ox_i \succ x_i \succ w \succ a \succ \cdots$. We refer to the former (resp. latter) set of voters as \emph{literal voters} for $x_i$ (resp. $\ox_i$).
\end{itemize}
The resulting instance has $m=4N+2$ candidates and $n\le2000N^2t$ voters.  By
the choice of $t$, we have $t\ge1000\cdot n^{1-\zeta}$ and $m\le n^\zeta$.

The key property of the construction is that the elimination order is stable
under fewer than $t/2$ vote changes.

\begin{proposition} \label{prop:elimination-order}
Let $\kappa$ be non-negative integer such that $\kappa < \frac{t}{2}$.
For any $\bpi'$ that is a $\kappa$-neighbor of $\bpi$, the elimination order for STV on $\bpi'$ must be as follows:
\begin{enumerate}
\item In the $j$-th round for $j \in [3N]$, let $i = \lceil j / 3 \rceil$ and $\ell_i \in \{x_i, \bar{x_i}\}$. The eliminations are as follows: \label{item:eliminiation-first-rounds}
\begin{itemize}
\item Case A: $j = 3i - 2$. $z_{\ell_i}$ is eliminated.
\item Case B: $j = 3i - 1$. $y_{\ell_i}$ is eliminated.
\item Case C: $j = 3i$. $z_{\oell_i}$ is eliminated (where $\oell_i$ is the negation of $\ell_i$).
\end{itemize}
\item In the $(3N + 1)$-th round, one of $w$ or $a$ is eliminated.
\item In the $(3N + 1 + i)$-th round for $i \in [N]$, $y_{\oell_i}$ is eliminated.
\item Finally, the remaining one among $w$ or $a$ is the winner.
\end{enumerate}
\end{proposition}

\begin{proof}
We start by proving item~\ref{item:eliminiation-first-rounds}. We prove this by strong induction. Assume that this is true for all rounds up to $j - 1$. Then, the three cases proceed as follows:
\begin{itemize}
\item Case A: $j = 3i - 2$. For both $\ell_i \in \{x_i, \ox_i\}$, the vote count of $z_{\ell_i}$ in $\bpi$ in this round is $(D + 5i)t$ (from Assignment Voters for $\ell_i$ and Selection Voters for $\ell_i$) whereas the vote counts for all other candidates are $\geq (D + 5i + 1)t$. Thus, the gap is $t > 2\kappa$. Since the vote count of each candidate w.r.t. $\bpi'$ differs that w.r.t. $\bpi$ by at most $\kappa < \frac{t}{2}$, we can conclude that one of them must be eliminated in this round.
\item Case B: $j = 3i - 1$. The vote count of $y_{\ell_i}$ is now at most $(D + 5i + 3)t$ (from Literal Voters for $\ell_i$, and possibly $2t$ Clause Voters) in $\bpi$. Meanwhile, the vote counts of all other candidates are at least $(D + 5i + 4)t$. (In particular, $z_{\oell_i}$ receives $4t$ additional votes from Selection Voters for $\ell_i$ in the previous round.) As a result, the gap is at least $t > \kappa / 2$ again, and we can conclude that $y_{\ell_i}$ must be eliminated in $\bpi'$.
\item Case C: $j = 3i$. In this round, $z_{\oell_i}$ has a vote count of $(D + 5i + 4)t$ while all other candidates have vote counts at least $(D + 5(i + 1))t$. Thus, we can again conclude that $z_{\oell_i}$ is eliminated.
\end{itemize}

Now, in the $(3N + 1)$-th round, we can see that all the remaining Literal Candidates have vote counts at least $(4D + 20i + 2)t$ (from all Assignment Voters, Selection Voters and Literal Voters of itself and its negation) in $\bpi$. Meanwhile, $a$ has at most $4Dt$ votes in $\bpi$. Thus, since the gap is more than $2\kappa$, the Literal Candidates cannot be eliminated in this round. Thus, either $a$ or $w$ must be eliminated.

Finally, once $a$ or $w$ is eliminated, the other candidate $o$ will have at least $8Dt$ votes in $\bpi$. Meanwhile, all the other Literal Candidates have only at most $(4D + 20i + 4)t$ vote counts in $\bpi$ before $o$ is eliminated. Since the gap between these two are more than $2\kappa$, we can once again conclude that all Literal Candidates are eliminated before $o$ and, thus, $o$ remains the winner.
\end{proof}

Given Proposition~\ref{prop:elimination-order}, it is simple to check that $w$ wins the election if the list of preferences is $\bpi$. Namely, after the first $3N$ rounds, $a$ will have at most $4Dt - 1$ votes while $w$ has at least $4Dt$ votes, so the former gets eliminated. Thus, $w$ is left as the winner. We will next prove the completeness and soundness of the reduction.

\paragraph{Completeness.} Suppose that there exists a satisfying assignment $\psi$ of $\Phi$. Let $\ell_1 , \dots, \ell_N$ denote the literals that are set to false by $\psi$. Fix $k \le \lfloor t/6 \rfloor$. (Recall that $\lfloor t/6 \rfloor \geq n^{1 - \zeta}$.)

We construct $\bpi'$ from $\bpi$ by keeping all the voters' preference except for $k$ of the Winner Voters where we change them to $z_{\oell_1} \succ \cdots \succ z_{\oell_N} \succ a \succ w \succ \cdots$. We claim that $a$ is the winner of the election and that it has MoV at least $\left\lceil \frac{k}{2} \right\rceil$. To see that this is the case, consider any $\bpi''$ that is a $\left(\left\lceil \frac{k}{2}\right\rceil - 1\right)$-neighbor of $\bpi'$. Note that this means that $\bpi''$ is a $\left(k + \left\lceil \frac{k}{2}\right\rceil - 1\right)$-neighbor of $\bpi$. Since $\left(k + \left\lceil \frac{k}{2}\right\rceil - 1\right) < \frac{t}{4}$, we may apply Proposition~\ref{prop:elimination-order}. %Thus, we know that, in the first $N$ rounds exactly one literal from each variable is eliminated. We claim that, for $i \in [N]$, $\ell_i$ is kept while its negation $\oell_i$ is eliminated in the first $N$ rounds. To see that this is the case, observe that

We claim that $z_{\ell_i}$ is eliminated in round $3i - 2$ for all $i \in [N]$. Again, this can be seen via induction: Suppose this holds in previous rounds. In the round $3i - 2$, due to the voters we changed, $z_{\oell_i}$ will have a vote count of $(D + 5i)t + k$ in $\bpi'$ while the vote count of $z_{\ell_i}$ remains at $(D + 5i)t$ in $\bpi'$. Since $\bpi''$ is a $\left(\left\lceil \frac{k}{2}\right\rceil - 1\right)$-neighbor of $\bpi'$, $z_{\ell_i}$ must also be eliminated in $\bpi''$. It also follows (from Proposition~\ref{prop:elimination-order}) that the corresponding $y_{\ell_i}$ is eliminated in round $3i - 1$.

Now, consider the $(3N + 1)$-th round. Since $\psi$ is a satisfying assignment, none of the Clause Voters votes for $w$ in this round. Thus, in $\bpi'$, $w$ only receives the vote counts from the $(4Dt - k)$ Winner Voters that were not changed. Meanwhile, $a$ receives votes from both the Alternate Voters and the $k$ Winner Voters that were adjusted, for a total count of $(4Dt - 1 + k)$ in $\bpi'$. Since this gap $(2k - 1)$ is greater than $2\left(\left\lceil \frac{k}{2}\right\rceil - 1\right)$, we can conclude that $w$ is eliminated in $\bpi''$. Thus, by Proposition~\ref{prop:elimination-order}, $a$ must be a winner in $\bpi''$.

\paragraph{Soundness.} Suppose that $\Phi$ admits no satisfying assignment. Consider any $\bpi'$ that is a $\kappa$-neighbor of $\bpi$ for any $\kappa < t/2$. From Proposition~\ref{prop:elimination-order},  the candidates eliminated in the first $3N$ rounds (for $\bpi'$) consists of $N$ literal variables, one from each variable $x_i$. Let $\psi$ denote the assignment that sets all of these literals to false. By our assumption, at least one clause is violated in $\Phi$. At round $(3N + 1)$, the corresponding Clause Voters will have $w$ as their most preferred candidate. Thus, $w$ has vote count of at least $4Dt + t$ in $\bpi$. Meanwhile, $a$ has vote count at most $4Dt$ in $\bpi$. Since $\kappa < t/2$ (and from Proposition~\ref{prop:elimination-order}), we can conclude that $a$ is eliminated in this round for $\bpi'$. Thus, by Proposition~\ref{prop:elimination-order}, we can conclude that $w$ remains the winner of the election. In other words, $\mov(\bpi, r) \geq \lceil t/2 \rceil \geq n^{1 - \zeta}$ as desired.

\section{Additional Details for Central Differential Privacy}
\label{app:missing-proof-section-3}

This appendix supplies the proofs omitted from
Section~\ref{sec:central-dp}.  We first analyze the utility of the generic
exponential-mechanism algorithm and its efficient Maximin variant.  We then
prove the computational barrier for private STV and conclude with the packing
lemmas used for the central-DP lower bounds.

\subsection{Upper Bounds}

We begin with the information-theoretic algorithm, whose score is the negative
margin of victory for every outcome other than the true winner.

\begin{proof}[Proof of \Cref{thm:central-em}]
For the score function in~\eqref{eqn:score}, the optimal utility is attained by
$r(\bpi)$ and equals
$\OPT(\bpi)=q(\bpi,r(\bpi))=0$,
and $|\cO|=m+1$.  Set
\[
L=\frac{2}{\eps}\log\!\left(\frac{m+1}{\gamma}\right)
\quad\text{and}\quad \tau=L+1.
\]
If $\mov(\bpi,r)\ge\tau$, then every $a\ne r(\bpi)$ satisfies
$q(\bpi,a)\le-\tau<-L$.  By Theorem~\ref{thm:expmech-utility}, with probability
at least $1-\gamma$ the exponential mechanism returns an outcome with utility
at least $-L$.  The only such outcome is $r(\bpi)$, including when
$r(\bpi)=\bot$.
\end{proof}

For Maximin, the efficiently computable score gap in~\eqref{eqn:2approx}
approximates the margin of victory within a factor of two, up to the one-vote
adjustment needed for fixed tie-breaking.  Indeed, changing one vote changes
each maximin score by at most $2$, so fewer than
$\Delta_{\mathrm{mm}}/4$ changes cannot alter the winner.  Conversely, the
construction in~\cite[Theorem~12]{Xia12} changes
$\Delta_{\mathrm{mm}}/2$ votes to make the original winner cease to be the
unique maximin-score maximizer.  From such a tie, one additional vote suffices
to change the winner selected by any fixed tie-breaking rule.  This proves
\eqref{eqn:2approx}.

\begin{proof}[Proof of \Cref{thm:cdp-upper-maximin}]
For $q'$, we have $\OPT(\bpi)=0$ and $|\cO|=m$.  Set
\[
L=\frac{2}{\eps}\log\!\left(\frac{m}{\gamma}\right)
\quad\text{and}\quad \tau=2L+2.
\]
By~\eqref{eqn:2approx}, if $\mov(\bpi,r)\ge\tau$, then every
$a\ne r(\bpi)$ satisfies
\[
q'(\bpi,a)=-\frac{\Delta_{\mathrm{mm}}}{4}
\le-\frac{\mov(\bpi,r)-1}{2}
\le-L-\frac12.
\]
Theorem~\ref{thm:expmech-utility} returns an outcome with score at least $-L$
with probability at least $1-\gamma$.  Hence it must return $r(\bpi)$, proving
the stated threshold.
\end{proof}

\subsection{Computational Barrier for STV}

The generic exponential mechanism is not efficient for STV because its score
requires computing the STV margin of victory.  The next proof strengthens this
observation: any efficient private algorithm satisfying the two usefulness
guarantees in Theorem~\ref{thm:stv-no-efficient-dp} would solve the hard gap
instance constructed in Theorem~\ref{thm:stv}.

\begin{proof}[Proof of Theorem~\ref{thm:stv-no-efficient-dp}]
Suppose for contradiction that such an algorithm $\alg$ exists.  We use it to
solve the distinguishing problem in Theorem~\ref{thm:stv} with
$\zeta=1/(C_2+1)$:
\begin{itemize}
\item Independently run $\alg(\bpi, r)$ for
$R=\lfloor m^{C_3}/3\rfloor$ repetitions.
\item If all the runs return $w$, then output NO.
\item Otherwise, output YES.
\end{itemize}

The running time is polynomial.  We analyze the two cases.
\begin{itemize}
\item (YES) Let
$k=\left\lceil4C_1\log(m)/\eps\right\rceil$.  For all sufficiently large
instances, $k\le n^{1-\zeta}$, so Theorem~\ref{thm:stv} gives a $k$-neighbor
$\bpi'$ of $\bpi$ such that $r(\bpi')=a$ and
\[
\mov(\bpi',r)\ge\left\lceil\frac{k}{2}\right\rceil
\ge\frac{C_1\log m}{\eps}.
\]
Usefulness therefore gives
$\Pr[\alg(\bpi',r)=a]\ge1-\gamma$.  By group privacy,
\begin{align*}
p:=\Pr[\alg(\bpi,r)=a]
&\ge e^{-\eps k}\Pr[\alg(\bpi',r)=a] \\
&\ge (1-\gamma)e^{-\eps}m^{-4C_1}.
\end{align*}
Since $C_3>4C_1$, for all sufficiently large $m$ this is at least
$5m^{-C_3}$.  Also, for sufficiently large $m$,
$R\ge m^{C_3}/4$.  Consequently, the probability of outputting YES is at least
\[
1-(1-p)^R\ge1-e^{-pR}\ge1-e^{-5/4}>\frac23.
\]
\item (NO) Since $m\le n^\zeta$ and $\zeta=1/(C_2+1)$, we have
$n^{1-\zeta}\ge m^{C_2}$.  Thus
$\mov(\bpi,r)\ge n^{1-\zeta}\ge m^{C_2}$.  Since $\alg$ is
$\left(m^{C_2},1/m^{C_3}\right)$-useful, each run outputs $w$ with probability
at least $1-1/m^{C_3}$.  By the union bound, the probability that any of the
$R\le m^{C_3}/3$ runs does not output $w$ is at most $1/3$.  Hence our
algorithm outputs NO with probability at least $2/3$.
\end{itemize}
Thus, we solve the distinguishing problem in Theorem~\ref{thm:stv} with
probability at least $2/3$ in polynomial time, implying NP $\subseteq$ BPP.
\end{proof}

\subsection{Packing Lower Bounds}

We next prove the generic packing lemma from
Section~\ref{sec:central-dp}.  The proof compares the probability of a fixed
candidate under two nearby profiles with different robust winners.

\begin{proof}[Proof of Lemma~\ref{lem:packing-generic}]
Let $\alg$ be an $\eps$-DP algorithm that is $(\tau,
\gamma)$-useful.  Suppose
for contradiction that
\[
\tau < T := \left\lfloor \frac{\ln(0.4m/\gamma)}{\kappa \eps} \right\rfloor.
\]
If $T=0$, the claimed lower bound is immediate, so assume $T\ge1$.
Let $\bpi_1, \dots, \bpi_m$ be profiles guaranteed by Assumption~\ref{assumption:packing} with $t = T$.

By the definition of usefulness, we must have
\begin{align*}
\Pr[\alg(\bpi_1) = c_1] \geq 1 - \gamma.
\end{align*}
Thus, there must exist $i \in [m] \setminus \{1\}$ such that
\begin{align*}
\Pr[\alg(\bpi_1) = c_i] \leq \frac{\gamma}{m - 1}.
\end{align*}
Since $\bpi_1, \bpi_i$ are $(\kappa T)$-neighbors, we have
\begin{align*}
\Pr[\alg(\bpi_i) = c_i]
&\leq e^{\eps \cdot (\kappa T)} \cdot \Pr[\alg(\bpi_1) = c_i] \\
&\leq e^{\eps \cdot (\kappa T)} \cdot \frac{\gamma}{m - 1}.
\end{align*}
However, usefulness on $\bpi_i$ requires this probability to be at least
$1-\gamma$.  Therefore
\begin{align*}
T \geq \frac{1}{\kappa \eps} \cdot \ln\left(\frac{(m - 1)(1 - \gamma)}{\gamma}\right) > \frac{\ln(0.4m/\gamma)}{\kappa \eps},
\end{align*}
where the last inequality uses $\gamma<0.2$.  This contradicts the definition of
$T$, so $\tau\ge T$.
\end{proof}

It remains to verify the packing assumption for the voting rules considered in
the paper.  We first handle every rule satisfying the majority criterion, then
give a separate symmetric construction for arbitrary non-trivial scoring
rules.

\subsubsection{Rules satisfying the majority criterion.} 

\begin{proof}[Proof of \Cref{lem:packing-majority}]
For each $i\in[m]$, let $\bpi_i$ be a profile with $2t$ voters, all of whom rank
$c_i$ first.  The majority criterion gives $r(\bpi_i)=c_i$; changing fewer than
$t$ voters cannot remove $c_i$'s strict majority; and any two such profiles are
$2t$-neighbors.  Thus Assumption~\ref{assumption:packing} holds with
$\kappa=2$.
\end{proof}

\subsubsection{Scoring rules.}

For scoring rules that need not satisfy the majority criterion, we start from a
cyclic profile in which every candidate has the same score and perturb it to
favor a prescribed candidate.

\begin{proof}[Proof of Lemma~\ref{lem:packing-scoring-based}]
Construct a symmetric base profile $\bpi_{\mathrm{cycle}}$ with exactly $m$
rankings.  For $1\le i\le m$, the $i$-th ranking is the cyclic shift
$c_i\succ c_{i+1}\succ\cdots\succ c_m\succ c_1\succ\cdots\succ c_{i-1}$.

By symmetry, every candidate appears once in every rank position.  Hence all
candidates have the same score in $\bpi_{\mathrm{cycle}}$:
\begin{align*}
\score_{\bs}(c; \bpi_{\mathrm{cycle}}) = s_1 + \cdots + s_m & & \forall c \in \cC.
\end{align*}

Let $\Delta = s_1 - s_m > 0$.

Let $\bpi_{\mathrm{base}}$ consist of $2t$ copies of
$\bpi_{\mathrm{cycle}}$.  This profile has $2mt$ voters and all candidates are
tied:
\begin{align} \label{eq:tie-base}
\score_{\bs}(c; \bpi_{\mathrm{base}}) = 2t\left(s_1 + \cdots + s_m\right) & & \forall c \in \cC.
\end{align}

For each candidate $c_i$, construct $\bpi_i$ from $\bpi_{\mathrm{base}}$ as
follows.  Identify the $2t$ rankings in which $c_i$ is last,
\[
x_1 \succ x_2 \succ \dots \succ x_{m-1} \succ c_i.
\]
and move $c_i$ to the first position in those rankings, shifting all other
candidates down by one position:
\[
c_i \succ x_1 \succ x_2 \succ \dots \succ x_{m-1}.
\]

In the modified profile $\bpi_i$, candidate $c_i$ gains $\Delta$ points in each
of the $2t$ changed rankings:
\[
\score_{\bs}(c_i; \bpi_i) = \score_{\bs}(c_i; \bpi_{\mathrm{base}}) + 2t\Delta.
\]

Every other candidate is shifted down in the modified rankings, so monotonicity
of the scoring vector implies that its score cannot increase:
\begin{align*}
\score_{\bs}(c_j; \bpi_i)
&\leq \score_{\bs}(c_j; \bpi_{\mathrm{base}}) \\
&\overset{\eqref{eq:tie-base}}{=} \score_{\bs}(c_i; \bpi_{\mathrm{base}})
&&\forall c_j \ne c_i.
\end{align*}

Thus $c_i$ wins in $\bpi_i$.  Moreover, its score advantage over every other
candidate is at least $2t\Delta$.  Changing one preference can decrease this gap
by at most $2\Delta$, so $\mov(\bpi_i,r)\ge t$.

Finally, $\bpi_i$ and $\bpi_j$ each differ from $\bpi_{\mathrm{base}}$ in $2t$
preferences, so they are $4t$-neighbors.
\end{proof}

\section{Additional Details for Non-Interactive LDP: Plurality}
\label{app:plurality-details}

This appendix proves the non-interactive LDP upper and lower bounds for
Plurality from Section~\ref{sec:noninteractive-ldp}.  For the upper bound, we
first establish simultaneous accuracy of the DJW-privatized score vector and
then translate score accuracy into winner recovery.  For the lower bound, we
reduce locally private two-point mean testing to winner publication in a
two-candidate Plurality election.

\subsection{Upper Bound}

\subsubsection{Private score release.}

Each voter represents their top-ranked candidate by a one-hot vector.  The DJW
randomizer privately releases this vector, and the aggregator sums the reports
to estimate every plurality score simultaneously.

\begin{algorithm}[H]
\caption{Local publication of DJW-privatized vote vectors \(\tilde{\bs}\) for Plurality}
\label{alg:publish-plurality-scores}
\begin{algorithmic}[1]
\REQUIRE Privacy parameter \(\varepsilon > 0\), candidate set \(C=\{c_1,\dots,c_m\}\)
\FORALL{voters \(i \in [n]\)}
    \STATE Let \(w_i\) be voter \(i\)'s top-ranked candidate
    \STATE Initialize the one-hot vector
    \(
        v_{i,j} \gets
        \begin{cases}
            1 & \text{if } c_j = w_i,\\
            0 & \text{otherwise}
        \end{cases}
    \)
    \STATE Release
    \(\tilde{\bv}_i\gets
    \mathcal{R}_{\mathrm{DJW}}(\bv_i;1,\varepsilon)\)
\ENDFOR
\STATE Aggregate the scores \(\tilde{\bs} \gets \sum_{i=1}^n \tilde{\bv}_i\)
\RETURN \(\tilde{\bs}\)
\end{algorithmic}
\end{algorithm}

\begin{corollary}
    \label{cor:acc-algo-plu}
    For $\eps>0$, $\gamma\in(0,1/2)$,
    \[
        \lambda=C_{\mathrm{DJW}}\kappa_\varepsilon
        \sqrt{2n\log(2m/\gamma)},
    \]
    and $\tilde{\bs}$ the output of
    Algorithm~\ref{alg:publish-plurality-scores},
    \[
        \Pr\!\left[\max_{j\in[m]} |\tilde{s}_j - \scoreplu(c_j)| > \lambda\right] \leq \gamma.
    \]
\end{corollary}

\begin{proof}
    Since $\|\bv_i\|_2=1$, Lemma~\ref{lem:djw-randomizer} implies that each
    release is $\varepsilon$-LDP and unbiased.  For each $c_j\in\cC$,
    the centered error
    \[
        \tilde{s}_j-\scoreplu(c_j)
        =\sum_{i=1}^n(\tilde{v}_{i,j}-v_{i,j})
    \]
    is a sum of independent sub-Gaussian random variables and is therefore
    sub-Gaussian with parameter at most
    $C_{\mathrm{DJW}}\kappa_\varepsilon\sqrt{n}$.  Thus
    \[
        \Pr\!\left[
        |\tilde{s}_j-\scoreplu(c_j)|>\lambda
        \right]
        \le 2\exp\!\left(
        -\frac{\lambda^2}
        {2C_{\mathrm{DJW}}^2\kappa_\varepsilon^2n}
        \right)
        =\frac{\gamma}{m}.
    \]
    A union bound over the $m$ candidates proves the claim.
\end{proof}

\subsubsection{Winner recovery.}

The preceding uniform error bound implies that an incorrect noisy winner can
occur only when the true plurality score gap, and hence the margin of victory,
is small.

\begin{proof}[Proof of Theorem~\ref{thm:ldp-plurality-upper}]
Let
$\lambda=C_{\mathrm{DJW}}\kappa_\varepsilon
\sqrt{2n\log(2m/\gamma)}$, and condition on the event from
Corollary~\ref{cor:acc-algo-plu}.  If the algorithm outputs $c\ne w$, then
\[
    \scoreplu(w)-\scoreplu(c)\le2\lambda.
\]
Lemma~\ref{lem:margin-plurality} gives
$\mov(\bpi,r)\le\lambda+1$.  Therefore every profile with
$\mov(\bpi,r)\ge\lambda+2$ is recovered correctly on the conditioned event,
which happens with probability at least $1-\gamma$.
\end{proof}

\subsection{Lower Bound}

The proof directly reduces two-point mean testing to publishing the winner of a
two-candidate Plurality election.

\begin{proof}[Proof of Theorem~\ref{thm:ldp-plurality-lower}]
Suppose such an algorithm $\alg$ exists.  Let $a>0$ be a sufficiently small
universal constant, and set
\[
    \theta=\frac{a}{\varepsilon}
    \sqrt{\frac{\log(1/\gamma)}{n}},
    \qquad
    \tau=\frac{n\theta}{4}
    =\frac{a}{4\varepsilon}\sqrt{n\log(1/\gamma)}.
\]
By choosing the constant $c$ in the theorem small enough, the assumption
$\gamma\ge e^{-cn\varepsilon^2}$ ensures that $\theta\le1/2$.

Let $P_+$ and $P_-$ be the distributions on $\{-1,1\}$ given by
\begin{align*}
    P_+(1)&=\frac{1+\theta}{2}, &
    P_+(-1)&=\frac{1-\theta}{2},\\
    P_-(1)&=\frac{1-\theta}{2}, &
    P_-(-1)&=\frac{1+\theta}{2}.
\end{align*}
Interpreting a sign $1$ as the ranking
$c_1\succ c_2$ and a sign $-1$ as $c_2\succ c_1$, $P_+$ has mean $\theta$ and favors candidate $c_1$, whereas $P_-$ has
mean $-\theta$ and favors candidate $c_2$.  We consider the following two hypotheses:
\[
    H_+:(X_1,\ldots,X_n)\sim P_+^{\otimes n},
    \qquad
    H_-:(X_1,\ldots,X_n)\sim P_-^{\otimes n}.
\]
We draw the votes according to one of these hypotheses, run $\alg$, and declare
$H_+$ if it returns $c_1$ and $H_-$ if it returns $c_2$.

Let $S$ be the sum of the signs.  Hoeffding's inequality gives
\begin{align*}
    \Pr_{P_+^{\otimes n}}\!\left[S<\frac{n\theta}{2}\right]
    &\le \exp\!\left(-\frac{n\theta^2}{8}\right),\\
    \Pr_{P_-^{\otimes n}}\!\left[S>-\frac{n\theta}{2}\right]
    &\le \exp\!\left(-\frac{n\theta^2}{8}\right)
    =\gamma^{a^2/(8\varepsilon^2)}.
\end{align*}
Set
\[
    \varepsilon_0=\min\left\{1,\frac{a}{\sqrt{8}}\right\}.
\]
Then $0<\varepsilon\le\varepsilon_0$ implies
$a^2/(8\varepsilon^2)\ge1$, and hence
$\gamma^{a^2/(8\varepsilon^2)}\le\gamma$.  On the complementary event the
Plurality winner agrees with the
hypothesis and the score gap is at least $n\theta/2$.  Since one vote change
alters this gap by at most $2$, the margin of victory is at least
$n\theta/4=\tau$.  Usefulness therefore makes the testing error under each
hypothesis at most $2\gamma$.

It remains to show that no locally private transcript permits such a test.
Let $Q_+$ and $Q_-$ be the transcript distributions under the two hypotheses.
The strong data-processing inequality for non-interactive local
privacy~\cite{DuchiJW18} gives
\[
    D_{\mathrm{KL}}(Q_+\|Q_-)
    \le Cn\varepsilon^2\theta^2
    =Ca^2\log(1/\gamma)
\]
for a universal constant $C$.  
Taking $a$ sufficiently small makes this less
than $\log(1/(8\gamma))$.

To spell out the testing step, let $A$ be the event that the test declares
$H_+$ (equivalently, that $\alg$ returns $c_1$), and let $A^c$ denote its
complement, the event that the test declares $H_-$ (equivalently, that $\alg$
returns $c_2$).  Taking $H_+$ as the null hypothesis, the type-I error is
declaring $H_-$ when $H_+$ is true, while the type-II error is declaring $H_+$
when $H_-$ is true.  Their probabilities are
\[
    e_{\mathrm{I}}=Q_+(A^c),
    \qquad
    e_{\mathrm{II}}=Q_-(A).
\]
Consequently,
\[
    e_{\mathrm{I}}+e_{\mathrm{II}}
    =1-\bigl(Q_+(A)-Q_-(A)\bigr)
    \ge 1-d_{\mathrm{TV}}(Q_+,Q_-).
\]
The Bretagnolle--Huber inequality gives
\[
    d_{\mathrm{TV}}(Q_+,Q_-)
    \le \sqrt{1-e^{-D_{\mathrm{KL}}(Q_+\|Q_-)}}.
\]
Since $1-\sqrt{1-x}\ge x/2$ for $x\in[0,1]$, these two displays imply
\[
    e_{\mathrm{I}}+e_{\mathrm{II}}
    \ge \frac12e^{-D_{\mathrm{KL}}(Q_+\|Q_-)}
    >4\gamma.
\]
But the preceding usefulness argument gave
$e_{\mathrm{I}},e_{\mathrm{II}}\le2\gamma$, and hence
$e_{\mathrm{I}}+e_{\mathrm{II}}\le4\gamma$, a contradiction.  Taking
$C'=a/4$ proves the theorem.
\end{proof}

\section{Additional Details for Non-Interactive LDP: Condorcet, Maximin, and Runoff}
\label{app:cmr-details}

This appendix fills in the non-interactive LDP details for Condorcet,
Maximin, and Plurality with Runoff from
Section~\ref{sec:noninteractive-ldp}.  We first derive all three upper bounds
from a common DJW release of pairwise margins.  We then develop a generic
marginal-release framework and instantiate it with a two-way profile generator
to obtain the lower bounds for Condorcet, Maximin, and Runoff.

\subsection{Upper Bounds}

\subsubsection{Pairwise-margin accuracy.}

All three algorithms use the same private pairwise statistics release, as described in \Cref{subsec:main-nonint-pairwise}.  The following
notation records the simultaneous coordinate error of the DJW release.

Let $q=\binom{m}{2}$ and, for privacy budget $\rho>0$ and failure probability
$\beta\in(0,1)$, define
\[
    \Lambda_{\mathrm{pair}}(\rho,\beta)
    :=C_{\mathrm{DJW}}\kappa_\rho
    \sqrt{2qn\ln(2q/\beta)}.
\]
By Lemma~\ref{lem:djw-randomizer} and a union bound over the $q$ pairwise
coordinates, every estimated pairwise margin is simultaneously accurate to
within $\Lambda_{\mathrm{pair}}(\rho,\beta)$ with probability at least
$1-\beta$.

\subsubsection{Condorcet.}

For Condorcet, an error either hides a true winner's weakest victory or creates
an apparent winner on a profile whose non-private outcome is $\bot$.

\begin{theorem}[Details for Theorem~\ref{thm:non-int-ldp-algo-pairwise}]
\label{thm:ldp-condorcet-upper}
    There exists a non-interactive $\varepsilon$-LDP algorithm that
    is $(\tau,\gamma)$-useful for Condorcet voting, where
    \[
        \tau=1+3\log m
        +3\Lambda_{\mathrm{pair}}(\varepsilon,\gamma).
    \]
    In particular, for $0<\varepsilon\le1$,
    \[
        \tau = O\left(\frac{m\sqrt{n\ln(m/\gamma)}}{\varepsilon}
        \right).
    \]
\end{theorem}

\begin{proof}
The algorithm is the pairwise DJW release and Condorcet aggregator described in
Section~\ref{subsec:main-nonint-pairwise}, so it is non-interactive and
$\varepsilon$-LDP.  Let
$\lambda=\Lambda_{\mathrm{pair}}(\varepsilon,\gamma)$ and condition on the event
that every estimated pairwise margin is within $\lambda$ of its true value.
The preceding concentration bound shows that this event has probability at
least $1-\gamma$.

If the true Condorcet outcome is a candidate $w$ and the algorithm does not
output $w$, then $w$ fails some noisy head-to-head comparison, so
$D_{\bpi}(w,d)\le \lambda$ for some $d\ne w$.  Lemma~\ref{lem:margin-condorcet}
then gives $\mov(\bpi,r)\le1+\lambda<\tau$.

If the true Condorcet outcome is $\bot$ and the algorithm outputs a candidate
$c$, then $\widetilde{D}(c,a)>0$ for every $a\ne c$, so
$D_{\bpi}(a,c)<\lambda$ for every $a\ne c$.  Hence
$\max_{a\ne c}D_{\bpi}(a,c)<\lambda$, and Lemma~\ref{lem:margin-condorcet} gives
$\mov(\bpi,r)<1+3\log m+3\lambda=\tau$.  Therefore, whenever
$\mov(\bpi,r)\ge\tau$, the algorithm returns the true Condorcet outcome with
probability at least $1-\gamma$.
\end{proof}

\subsubsection{Maximin.}

For Maximin, uniform pairwise accuracy directly controls every candidate's
noisy maximin score.

\begin{theorem}[Details for Theorem~\ref{thm:non-int-ldp-algo-pairwise}]
\label{thm:maximin_ldp}
    There exists a non-interactive $\varepsilon$-LDP algorithm that is
    $(\tau, \gamma)$-useful for the Maximin voting rule, where
    $\tau=\Lambda_{\mathrm{pair}}(\varepsilon,\gamma)+2$.
    In particular, for $0<\varepsilon\le1$,
    \[
        \tau = O\left(
        \frac{m\sqrt{n\ln(m/\gamma)}}{\varepsilon}
        \right).
    \]
\end{theorem}

\begin{proof}
The algorithm is the pairwise DJW release and Maximin aggregator described in
Section~\ref{subsec:main-nonint-pairwise}, so it is non-interactive and
$\varepsilon$-LDP.  By the concentration bound above, with probability at
least $1-\gamma$,
\[
    \max_{a\ne b}|\widetilde{D}(a,b)-D_{\bpi}(a,b)|
    \le \lambda_{\mathrm{mm}}
    :=\Lambda_{\mathrm{pair}}(\varepsilon,\gamma).
\]
Condition on this event.  Recall that $\scmaximin(a)$ and $\tscmaximin(a)$ denote the true and
noisy maximin scores of $a$, respectively.  Since each pairwise margin is
accurate to within $\lambda_{\mathrm{mm}}$, every maximin score is also accurate
to within $\lambda_{\mathrm{mm}}$.

The noisy-error consequence in Lemma~\ref{lem:margin-maximin} shows that an
erroneous output requires
$\mov(\bpi,r)\le\lambda_{\mathrm{mm}}+1<\tau$.  Hence, whenever
$\mov(\bpi,r)\ge\tau$, the algorithm cannot err on the conditioned event, and
it returns the true winner with probability at least $1-\gamma$.
\end{proof}

\subsubsection{Plurality with Runoff.}

Algorithm \ref{alg:noninteractive-plu-runoff} combines the private plurality release from
Appendix~\ref{app:plurality-details} with the common pairwise release above.
Its aggregator is the same noisy runoff rule as in the interactive protocol;
the difference is that all pairwise margins are released before the finalists
are known.

\begin{theorem}[Details for Theorem~\ref{thm:non-int-ldp-algo-pairwise}]
\label{thm:ldp-runoff-upper}
    There exists a non-interactive $\varepsilon$-LDP algorithm that is
    $(\tau,\gamma)$-useful for Plurality with Runoff, where
    \begin{align*}
        \tau
        &=3+\max\left\{
        2\lambda_{\mathrm{plu}},
        \lambda_{\mathrm{plu}}+\frac{\lambda_{\mathrm{pair}}}{2}
        \right\}, \\
        \lambda_{\mathrm{plu}}
        &= C_{\mathrm{DJW}}\kappa_{\varepsilon/2}
        \sqrt{2n\ln(4m/\gamma)}, \\
        \lambda_{\mathrm{pair}}
        &=\Lambda_{\mathrm{pair}}(\varepsilon/2,\gamma/2).
    \end{align*}
    In particular, for $0<\varepsilon\le1$,
    \[
        \tau = O\left(
        \frac{m\sqrt{n\ln(m/\gamma)}}{\varepsilon}
        \right).
    \]
\end{theorem}

\begin{algorithm}[t]
    \caption{Non-Interactive Algorithm for Plurality with Runoff}
    \label{alg:noninteractive-plu-runoff}
    \begin{algorithmic}[1]
        \REQUIRE Preference profile $\bpi=(\pi_1,\ldots,\pi_n)$, privacy parameter $\varepsilon$
        \ENSURE Winner $c \in \cC$
        \STATE Each voter releases a DJW-privatized plurality vector using Algorithm~\ref{alg:publish-plurality-scores} with privacy budget $\varepsilon/2$
        \STATE Each voter releases a DJW-randomized pairwise vector using the construction from Section~\ref{subsec:main-nonint-pairwise} with privacy budget $\varepsilon/2$
        \STATE Aggregate the noisy plurality scores $\tilde{v}$ and the noisy pairwise margins $\widetilde{D}$
        \STATE Let $a,b$ be the two candidates with the largest coordinates of $\tilde{v}$
        \IF{$\widetilde{D}(a,b)>0$}
            \RETURN $a$
        \ELSE
            \RETURN $b$
        \ENDIF
    \end{algorithmic}
\end{algorithm}

\begin{proof}
Algorithm~\ref{alg:noninteractive-plu-runoff} is non-interactive because each
voter applies the same combined randomizer once. The noisy plurality release is
$\varepsilon/2$-LDP, and so is the DJW pairwise release. By composition,
the algorithm is $\varepsilon$-LDP.

By Corollary~\ref{cor:acc-algo-plu}, with privacy budget $\varepsilon/2$ and
failure probability $\gamma/2$, all noisy plurality scores are within
$\lambda_{\mathrm{plu}}$ of their true values with probability at least
$1-\gamma/2$. By the DJW concentration bound, with failure probability
$\gamma/2$, every noisy pairwise margin is within
$\lambda_{\mathrm{pair}}$ of its true value. Thus, with probability at least
$1-\gamma$, both events hold. We condition on this event.

Lemma~\ref{lem:margin-runoff}, with
$\eta=\lambda_{\mathrm{plu}}$ and
$\zeta=\lambda_{\mathrm{pair}}$, shows that an erroneous output requires
\[
    \mov(\bpi,r)
    \le2+\max\left\{
    2\lambda_{\mathrm{plu}},
    \lambda_{\mathrm{plu}}+\frac{\lambda_{\mathrm{pair}}}{2}
    \right\}
    <\tau.
\]
Therefore the algorithm outputs the true winner with probability at least
$1-\gamma$ whenever $\mov(\bpi,r)\ge\tau$.
\end{proof}

\subsection{Generic Marginal-Release Framework}
\label{sec:noninteractive-ldp-lower-reductions}

Beginning with this subsection, we develop the non-interactive lower bounds
for Condorcet, Maximin, and Plurality with Runoff.

For a Boolean dataset $X=(x_1,\ldots,x_N)$ with $x_i\in\{0,1\}^d$, recall from Definition~\ref{def:k-way-marginal} the $k$-way marginal vector $Q(X)=(Q_A(X))_{A\in\binom{[d]}{k}}$. We denote by $\bar Q(X)=Q(X)/N$ the vector of normalized marginals.

We start by stating the following lower bounds for releasing normalized $k$-way marginals from~\cite{edmonds2020thepower}.

\begin{theorem}[\citealt{edmonds2020thepower}]
%[Local marginal lower bound~\cite{edmonds2020thepower}]
\label{thm:edmonds-local-marginal-lower}
For
every non-interactive $\varepsilon$-LDP mechanism $\cM$, with
$0<\varepsilon\le1$, releasing all normalized marginals, %from
%Definition~\ref{def:k-way-marginal}, 
the following minimax lower bounds hold,
up to logarithmic factors:
\begin{enumerate}\raggedright
    \item For two-way marginals\footnote{This lower bound follows from two-way parity queries, which can be simulated using two-way marginals by querying both the original and bit-flipped inputs. Combining the two marginals yields the count of inputs with $x_i\oplus x_j=0$, recovering parity with at most a factor-of-two increase in the privacy and error parameters.}, 
    \[
        \begin{aligned}
        \sup_X\mathbb{E}\|\cM(X)-\bar Q(X)\|_\infty
        \ge\widetilde{\Omega}\!\left(
        \min\left\{1,\frac{d}{\varepsilon\sqrt{N}}\right\}\right).
        \end{aligned}
    \]
    \item More generally, for $k$-way marginals with $1\le k\le d/2$,
    \[
        \begin{aligned}
        &\sup_X\mathbb{E}\|\cM(X)-\bar Q(X)\|_\infty\\
        &\qquad\ge\widetilde{\Omega}\!\left(\min\left\{1,
        \frac{(d/k)^{\Omega(\sqrt{k})}}{\varepsilon\sqrt{N}}
        \right\}\right).
        \end{aligned}
    \]
\end{enumerate}
\end{theorem}

Our lower bounds reduce private marginal release to private winner
publication.  We first isolate the properties required of a rule-specific
profile generator, then analyze a nonadaptive threshold scan that converts
useful winner predictions into accurate marginal estimates.

\subsubsection{Compatible profile generators.}

%The input is a database
%$X=(x_1,\ldots,x_N)$ with $x_i\in\{0,1\}^d$, and the target is the vector of all
%$k$-way conjunction marginals. 
The only rule-specific ingredient in our reduction is a \emph{profile
generator} that encodes each query and threshold as an election.

Fix a voting rule $r$ and a non-interactive LDP algorithm $\alg_r$.  On a
profile $\bpi=(\pi_1,\ldots,\pi_n)$, the algorithm applies its  randomizers\footnote{This notation includes the special
case in which every voter uses the same randomizer.}
$\cR^r_{1,\varepsilon},\ldots,\cR^r_{n,\varepsilon}$ on $\pi_1, \dots, \pi_n$ respectively, and then runs its
aggregator on the privatized messages. 

The profile generator
$\cF$ has four components:
\begin{itemize}\raggedright
    \item $\cF_\mathrm{cand}$ constructs the candidate set $\cC$;
    \item $\cF_\mathrm{data}$ constructs the \emph{data-dependent profile} $\bpi_X$ where each ranking in $\bpi_X$ depends on exactly a single $x_i$, and the per-ranking privacy budget $\varepsilon'$;
    \item $\cF_\mathrm{query}$ constructs a \emph{query profile} $\bpi_A$ and designated winner $w_A$;
    \item $\cF_\mathrm{thres}$ constructs a \emph{threshold profile} $\bpi_A^{(\tau)}$.
\end{itemize}
For query $A$ and threshold $\tau$, the election profile is
$\bpi_X\circ\bpi_A\circ\bpi_A^{(\tau)}$.

The generator is useful when the winner tests the marginal value
$\mu_A=\sum_{i=1}^N\prod_{j\in A}x_{i,j}$.  The following definition records the
four properties needed by the reduction.
\begin{definition}[Compatibility]
    \label{def:compatibility}
    A non-interactive LDP voting algorithm for $r$ and a profile generator
    $\cF$ are \emph{compatible at target budget $\varepsilon$} if, for every
    threshold $\tau\in[N]$, the output
    $(\cC,\bpi^{(\tau)},w_A,\varepsilon')\gets\cF(d,X,A,\tau,
    \varepsilon)$ satisfies:
    \begin{enumerate}\raggedright
        \item[(i)] \label{cond:compat1} $w_A=r(\bpi^{(\tau)})$ if and only if $\sum_{i=1}^N\prod_{j\in A}x_{i,j}\ge\tau$;
        \item[(ii)] \label{cond:compat2} $\mov(\bpi^{(\tau)},r)\ge\left|\sum_{i=1}^N\prod_{j\in A}x_{i,j}-\tau\right|$;
        \item[(iii)] every $\bpi^{(\tau)}$ has the same number of rankings,
        independent of $A$ and $\tau$;
        \item[(iv)] \label{cond:compat4} $s\varepsilon'=\varepsilon$, where $s$
        is the maximum number of data-dependent rankings that depend on any
        one record $x_i$.
    \end{enumerate}
\end{definition}

\subsubsection{Non-adaptive threshold scan.}

Compatibility lets us recover $\mu_A$ by testing all thresholds.  It is
important that these tests are fixed in advance: because they reuse the same
privatized copy of $\bpi_X$, their errors need not be independent, so an
adaptive argument cannot treat them as independent trials.

Algorithm~\ref{alg:k-way_reduction} first privatizes the data-dependent
profile $\bpi_X$ once and reuses those messages in every test.  For each
marginal query $A$ and threshold $\tau$, it appends the public query and
threshold profiles produced by $\cF$, generates their messages, and runs the
voting aggregator.  Compatibility makes the event that the aggregator returns
$w_A$ a test of whether $\mu_A\ge\tau$.  The indicator $T[\tau]$ records the
result of this test, and the largest accepted threshold is used as the estimate
$\hat\mu_A$.  Only the reused messages for $\bpi_X$ depend on the private
dataset; the remaining messages are generated from public rankings and incur
no additional privacy loss. 

\begin{algorithm}[h]
    \caption{Non-Adaptive Reduction to $\varepsilon$-LDP $k$-Way Marginals}
    \label{alg:k-way_reduction}
    {\raggedright
    \textbf{Input:} Data $x_1,\ldots,x_N$ where $x_i\in\{0,1\}^d$ for all $i\in [N]$ \\
    \textbf{LDP Algorithm:} $\alg_r$ with randomizers $(\cR_i^r)_{i\in[n]}$ and aggregator $\mathcal{A}^r$ \\
    \textbf{Profile Generator:} $\cF$ composed of components $\cF_\mathrm{cand},\cF_\mathrm{data},\cF_\mathrm{query},\cF_\mathrm{thres}$.
    \par}

    \begin{algorithmic}[1]
        \STATE $\cC\gets \cF_\mathrm{cand}(d)$.
        \STATE $(\bpi_X,\varepsilon')\gets \cF_\mathrm{data}(X,\varepsilon)$.
        \STATE Let $n_X=|\bpi_X|$ and privatize its rankings once using the corresponding randomizers $\cR_{1,\varepsilon'}^r,\ldots,\cR_{n_X,\varepsilon'}^r$.
        \FORALL{$A\subseteq [d]$ such that $|A|=k$}
            \STATE $(\bpi_{A},w_A)\gets \cF_\mathrm{query}(A)$.
            \STATE Let $T$ be an array of size $N$.
            \FORALL{$\tau\in [N]$}
                \STATE $\bpi^{(\tau)}_A\gets \cF_\mathrm{thres}(A,w_A,\tau)$.
                \STATE Randomize the public rankings in $\bpi_A\circ\bpi_A^{(\tau)}$ using the corresponding remaining randomizers $\cR_{n_X+1,\varepsilon'}^r,\ldots,\cR_{n,\varepsilon'}^r$.
                \STATE Run $\mathcal{A}^r$ on the reused messages for $\bpi_X$ and the public messages for this $(A,\tau)$, obtaining $w$.
                \STATE $T[\tau]\gets \mathbf{1}(w=w_A)$.
            \ENDFOR
            \STATE $\hat\mu_A\gets\max(\{0\}\cup\{\tau\in[N]:T[\tau]=1\})$.
            \STATE Output $\hat{\mu}_A$ for marginal query $A$.
        \ENDFOR
    \end{algorithmic}
\end{algorithm}

The next lemma bounds both the high-probability and expected errors.  The
expected-error statement is useful because the marginal lower bounds invoked
below are formulated in expectation.

\begin{lemma}
\label{lem:k-way-reduction}
Let $X$ be a dataset and let $\alg_r$ be an $\varepsilon'$-LDP algorithm that
is $(\tau,\gamma)$-useful for voting rule $r$.  Let $Q(X)$ be the vector of all
$k$-way marginal answers.  If there is a profile generator $\cF$ compatible
with $\alg_r$ at target budget $\varepsilon$, and each record contributes at
most $s$ data-dependent rankings, then
$\varepsilon'=\varepsilon/s$ by condition~\ref{cond:compat4}, and there is an
$\varepsilon$-LDP algorithm $\cM$ such that
\[
    \Pr\!\left[\|\cM(X)-Q(X)\|_\infty>\tau+1\right]
    \le N\binom{d}{k}\gamma.
\]
Moreover, for the normalized answers $Q(X)/N$,
\[
    \mathbb{E}\!\left[
    \left\|\frac{\cM(X)}{N}-\frac{Q(X)}{N}\right\|_\infty
    \right]
    \le \frac{\tau+1}{N}+N\binom{d}{k}\gamma.
\]
\end{lemma}
\begin{proof}
The data-dependent messages are generated only once.  Every record contributes
at most $s$ messages, each with privacy budget $\varepsilon'$, so condition
\ref{cond:compat4} (iv) and composition theorem make their joint release
$\varepsilon$-LDP.  All remaining rankings are public, and every subsequent
operation is post-processing.  Thus $\cM$ is $\varepsilon$-LDP.

Fix $A$ and write $\mu_A=Q_A(X)$.  For every fixed threshold $\tau$ satisfying
$|\mu_A-\tau|\ge\alpha$, compatibility and usefulness imply that $T[\tau]$ is
correct with probability at least $1-\gamma$.  There are at most
$N\binom{d}{k}$ fixed query--threshold pairs.  A union bound, which does not
require independence, shows that all such tests are simultaneously correct
with probability at least $1-N\binom{d}{k}\gamma$.

Condition on this event.  No threshold greater than $\mu_A+\alpha$ is accepted,
so $\hat\mu_A\le\mu_A+\tau$.  Conversely, if
$\mu_A>\tau+1$, the integer
$\tau_A=\lfloor\mu_A-\tau\rfloor$ is accepted and satisfies
$\tau_A>\mu_A-\tau-1$; if $\mu_A\le\tau+1$, the default value $0$ already
satisfies the required lower bound.  Thus
$|\hat\mu_A-\mu_A|\le\tau+1$ for every $A$ (the harmless $+1$ also covers
non-integral $\tau$ and boundary ties).

The normalized error is at most $1$ on the complementary event.  Splitting the
expectation according to the good event proves the second inequality.
\end{proof}

\subsection{Two-Way Reduction for Condorcet, Maximin, and Plurality with Runoff}

We now instantiate the generic reduction for two-way marginals.  A single family
of profile generators suffices for all three rules.  The padding forces the two
queried candidates to dominate all outsiders and keeps the election size fixed
for every threshold.

\subsubsection{Profile generator.}

For a dataset $X=(x_1,
\ldots,x_N)$ with $x_i\in\{0,1\}^d$, a query
$A=\{j,k\}$, and privacy budget $\varepsilon$, define the generator
$\cF_\mathrm{2way}^{(\nu)}$ as follows.  Assume, without loss of generality, that
$c_j\succ_\sigma c_k$ in the fixed ranking $\sigma$.
\begin{itemize}\raggedright
    \item $\cF_\mathrm{cand}(d)$ returns the candidate set
    $\cC=\{c_1,\ldots,c_d\}$.  Let $\sigma$ be a fixed ranking of $\cC$, and let
    $\sigma^R$ be its reverse.
    \item $\cF_\mathrm{data}(X,\varepsilon)$ returns
    $(\bpi_X,\varepsilon')$, where
    $\bpi_X=(\pi_1,\ldots,\pi_{2N})$ and $\varepsilon'=\varepsilon/2$.
    For each record $x_i$, let
    $S_i=\{c_j:x_{i,j}=1\}$ and $\bar S_i=\cC\setminus S_i$, and define
    \begin{align*}
        \pi_{2i-1}:&\quad (\sigma\mid_{S_i})\succ(\sigma\mid_{\bar S_i}), \\
        \pi_{2i}:&\quad (\sigma^R\mid_{\bar S_i})\succ(\sigma\mid_{S_i}).
    \end{align*}
    Exactly two rankings are associated with each data point, so
    $\varepsilon'=\varepsilon/2$ satisfies condition~\ref{cond:compat4} (iv).
    \item $\cF_\mathrm{query}(A)$ returns $(\bpi_A,w_A)$, where $w_A=c_j$ and
    $\bpi_A=(\pi'_1,\ldots,
    \pi'_{4N+2\nu})$.  For each $i\in[2N+\nu]$, set
    \begin{align*}
        \pi'_{2i-1}:&\quad c_j\succ c_k\succ
        (\sigma\mid_{\cC\setminus\{c_j,c_k\}}), \\
        \pi'_{2i}:&\quad c_k\succ c_j\succ
        (\sigma\mid_{\cC\setminus\{c_j,c_k\}}).
    \end{align*}
    \item $\cF_\mathrm{thres}(A,w_A,\tau)$ returns the fixed-size profile
    \[
        \bpi_A^{(\tau)}
        =(\pi_A^+)^{N-\tau}\circ(\pi_A^-)^{N+\tau-1},
    \]
    where
    \begin{align*}
        \pi_A^+:&\quad
        c_j\succ c_k\succ(\sigma\mid_{\cC\setminus\{c_j,c_k\}}),\\
        \pi_A^-:&\quad
        c_k\succ c_j\succ(\sigma\mid_{\cC\setminus\{c_j,c_k\}}).
    \end{align*}
    This profile has exactly $2N-1$ rankings for every $\tau\in[N]$.
\end{itemize}

\subsubsection{Pairwise-margin calculation.}

The two rankings associated with each data point are designed so that their
net comparison between the queried candidates records exactly the corresponding
two-way conjunction.

\begin{lemma}
For any $\nu\in\bbZ_{\ge0}$, let
$(\bpi_X,\cdot)=\cF_\mathrm{data}(X,\cdot)$ be the data profile in
$\cF_\mathrm{2way}^{(\nu)}$.  For any $j,k\in[d]$ with
$c_j\succ_\sigma c_k$,
\begin{align*}
    N_{\bpi_X}(c_j\succ c_k)
    &=2\mu_{11}+\mu_{10}+\mu_{01}+\mu_{00}, \\
    N_{\bpi_X}(c_k\succ c_j)
    &=\mu_{10}+\mu_{01}+\mu_{00},
\end{align*}
where $\mu_{ab}=|\{i\in[N]:x_{i,j}=a \wedge x_{i,k}=b\}|$.
\end{lemma}
\begin{proof}
The two rankings associated with record $x_i$ compare $c_j$ and $c_k$ as follows:
\begin{center}
\begin{tabular}{c|c|c|c}
    $x_{i,j}$ & $x_{i,k}$ & $\pi_{2i-1}$ & $\pi_{2i}$  \\
    \hline
    $1$ & $1$ & $\succ$ & $\succ$ \\
    $1$ & $0$ & $\succ$ & $\prec$ \\
    $0$ & $1$ & $\prec$ & $\succ$ \\
    $0$ & $0$ & $\succ$ & $\prec$
\end{tabular}
\end{center}
Here $\succ$ means $c_j\succ c_k$, and $\prec$ means $c_k\succ c_j$.  Counting the
rows of the table gives the two displayed identities.
\end{proof}

The query profile contributes equally to the two pairwise directions, while the
threshold profile gives $c_k$ a net pairwise advantage of $2\tau-1$ over $c_j$.
Therefore we have the following consequence.
\begin{corollary}
\label{cor:final_round}
For any $\nu\in\bbZ_{\ge0}$, let $\bpi_\mathrm{2way}$ be a profile generated by
$\cF_\mathrm{2way}^{(\nu)}$.  For any query $A=\{j,k\}$ with
$c_j\succ_\sigma c_k$,
\[
    N_{\bpi_\mathrm{2way}}(c_j\succ c_k)
    -N_{\bpi_\mathrm{2way}}(c_k\succ c_j)
    =2\mu_{11}-2\tau+1,
\]
where $\mu_{11}=|\{i\in[N]:x_{i,j}=x_{i,k}=1\}|$.
\end{corollary}
\begin{proof}
The generated profile is the concatenation
\[
    \bpi_\mathrm{2way}
    =\bpi_X\circ\bpi_A\circ\bpi_A^{(\tau)}.
\]
By the preceding lemma, the data profile contributes pairwise margin
\[
    N_{\bpi_X}(c_j\succ c_k)-N_{\bpi_X}(c_k\succ c_j)
    =2\mu_{11}.
\]
The query profile consists of $2N+\nu$ rankings in each pairwise direction,
so its contribution to the margin is zero.  Finally, the threshold profile
contains $N-\tau$ rankings with $c_j\succ c_k$ and $N+\tau-1$ rankings with
$c_k\succ c_j$, and therefore contributes
\[
    (N-\tau)-(N+\tau-1)=1-2\tau.
\]
Adding the three contributions gives
\[
    N_{\bpi_\mathrm{2way}}(c_j\succ c_k)
    -N_{\bpi_\mathrm{2way}}(c_k\succ c_j)
    =2\mu_{11}-2\tau+1,
\]
as claimed.
\end{proof}

Thus, whenever $c_j$ and $c_k$ are the decisive pair for the voting rule, the
designated candidate $c_j$ wins exactly when
$\sum_{i=1}^Nx_{i,j}x_{i,k}\ge\tau$.

\subsubsection{Compatibility with the voting rules.}

We now verify the remaining compatibility conditions.  For Condorcet and Maximin, the query
profile makes the queried pair dominate every outside candidate, so their
head-to-head comparison determines the winner.

\begin{lemma}[Condorcet and Maximin]
The profile generator $\cF_\mathrm{2way}^{(0)}$ is compatible with the
non-interactive Condorcet and Maximin algorithms.
% \phanu{$\cF_\mathrm{2way}^{(0)}$ suffices.}
\end{lemma}
\begin{proof}
Fix a query $A=\{j,k\}$ with $c_j\succ_\sigma c_k$, and write
$\mu=\sum_{i=1}^Nx_{i,j}x_{i,k}$.  For every
$\ell\in[d]\setminus A$, the query profile ensures
\begin{align*}
    N_{\bpi_\mathrm{2way}}(c_j\succ c_\ell)
    -N_{\bpi_\mathrm{2way}}(c_\ell\succ c_j)
    &\ge 2N+2\nu, \\
    N_{\bpi_\mathrm{2way}}(c_k\succ c_\ell)
    -N_{\bpi_\mathrm{2way}}(c_\ell\succ c_k)
    &\ge 2N+2\nu.
\end{align*}
Thus both queried candidates beat every outside candidate head-to-head.  By
Corollary~\ref{cor:final_round}, $c_j$ beats $c_k$ exactly when $\mu\ge\tau$;
otherwise $c_k$ beats $c_j$.  Hence one of $c_j,c_k$ is the Condorcet winner,
so both Condorcet and Maximin return $c_j$ exactly when $\mu\ge\tau$.

For completeness, set $r=|\mu-\tau|$.  The decisive pairwise margin has
absolute value $2r+1$ when $\mu\ge\tau$ and $2r-1$ otherwise.  Moreover, the
displayed margins against every outside candidate are at least $2N$ even when
$\nu=0$.  Because $r\le N$, changing fewer than $r$ rankings preserves both
the sign of the decisive comparison and the victories of the queried
candidates over every outsider.  The same queried candidate therefore remains
the Condorcet and Maximin winner.  Thus the margin of victory is at least
$|\mu-\tau|$, and the compatibility conditions hold.
\end{proof}

For Plurality with Runoff, additional padding forces the same queried pair to
be the two first-round finalists.

\begin{lemma}[Plurality with Runoff]
For Plurality with Runoff, $\cF_\mathrm{2way}^{(2N)}$ is compatible with
$\alg_\mathrm{Plu2}$.
\end{lemma}
\begin{proof}
Fix a query $A=\{j,k\}$ with $c_j\succ_\sigma c_k$, and write
$\mu=\sum_{i=1}^Nx_{i,j}x_{i,k}$.  In the generated profile,
\begin{align*}
    \scoreplu(c_a)&\ge 2N+\nu &&\text{for each }a\in A, \\
    \scoreplu(c_\ell)&\le 2N &&\text{for each }\ell\in[d]\setminus A.
\end{align*}
With $\nu=2N>0$, both $c_j$ and $c_k$ are therefore plurality finalists.  Their
runoff is decided by the pairwise margin in Corollary~\ref{cor:final_round}, so
Plurality with Runoff returns $c_j$ exactly when $\mu\ge\tau$.  Changing fewer
than $\min\{\nu/2,|\mu-\tau|\}$ rankings cannot either remove one of the queried
candidates from the runoff or reverse the final pairwise comparison.  With
$\nu=2N$, the margin is at least $|\mu-\tau|$, and compatibility follows.
\end{proof}

\subsubsection{Reduction and lower bound.}

The two compatibility lemmas let us invoke the generic marginal-release
framework.  We first state the resulting two-way reduction and then combine it
with the local marginal lower bound.

\begin{theorem}[Reduction for Corollary~\ref{cor:ldp-cmr-main-lower}]
\label{thm:ldp-two-way-reduction}
Let $X=(x_1,\ldots,x_N)$ with $x_i\in\{0,1\}^d$, and let
$Q(X)=(Q_A(X))_{A\in\binom{[d]}{2}}$, where
$Q_A(X)=\sum_{i=1}^N\prod_{j\in A}x_{i,j}$.  If there exists an
$\varepsilon/2$-LDP algorithm for Condorcet, Maximin, or Plurality with
Runoff---with $d$ candidates and $12N-1$ voters---that is
$(\alpha,\gamma)$-useful, then
there exists an $\varepsilon$-LDP algorithm $\cM$ such that
\[
    \Pr\!\left[\|\cM(X)-Q(X)\|_\infty>\alpha+1\right]
    \le N\binom{d}{2}\gamma,
\]
and
\[
    \mathbb{E}\!\left[
    \left\|\frac{\cM(X)}{N}-\frac{Q(X)}{N}\right\|_\infty
    \right]
    \le \frac{\alpha+1}{N}+N\binom{d}{2}\gamma.
\]
\end{theorem}
\begin{proof}
Apply Lemma~\ref{lem:k-way-reduction} with $k=2$ and use
$\cF_\mathrm{2way}^{(2N)}$ for all three rules.  Every generated election has
\[
    2N+(4N+4N)+(2N-1)=12N-1
\]
voters.  Each data point contributes exactly two rankings, each privatized with
budget $\varepsilon/2$, so the induced marginal-release algorithm is
$\varepsilon$-LDP.  Lemma~\ref{lem:k-way-reduction} gives the stated bounds.
\end{proof}

Combining Theorem~\ref{thm:ldp-two-way-reduction} with
Theorem~\ref{thm:edmonds-local-marginal-lower} gives the candidate dependence
used in the summary table.

\begin{proof}[Proof of Corollary~\ref{cor:ldp-cmr-main-lower}]
Set $\varepsilon_0=1/2$ and apply
Theorem~\ref{thm:ldp-two-way-reduction} with its privacy parameter set to
$2\varepsilon$, as well as $d=m$, $N=\Theta(n)$, and $\alpha=\tau$.  The
assumed voting algorithm is then the required $(2\varepsilon)/2$-LDP
algorithm, and the resulting marginal mechanism is $2\varepsilon$-LDP.  Since
$2\varepsilon\le1$, Theorem~\ref{thm:edmonds-local-marginal-lower} applies; the
factor of two is absorbed into the constants.  If
$\gamma\le c/(m^2n^2)$ for a sufficiently small $c$ and
$\tau=\widetilde{o}(\min\{n,m\sqrt{n}/\varepsilon\})$, the expected normalized error in
Theorem~\ref{thm:ldp-two-way-reduction} would be
$\widetilde{o}(\min\{1,m/(\varepsilon\sqrt{N})\})$: the first term has this
order, and $N\binom{m}{2}\gamma=O(1/N)$.  This contradicts the two-way case of
Theorem~\ref{thm:edmonds-local-marginal-lower}.  Hence the displayed lower
bound on $\tau$ follows.
\end{proof}

Corollary~\ref{cor:ldp-cmr-main-lower} explains the
$\widetilde{\Omega}(m\sqrt{n}/\varepsilon)$ non-interactive LDP lower-bound entries for
Condorcet, Maximin, and Plurality with Runoff in
Table~\ref{tab:results-summary}.

\section{Additional Details for Non-Interactive LDP: STV}
\label{app:stv-details}

The non-interactive STV algorithm uses the same noisy elimination rule as the
interactive protocol, but releases all statistics that might be needed before
the elimination process begins.  Recall that
\[
    \mathcal{I}
    =\{(S,c):S\subseteq\cC,\ |S|\ge2,\ c\in S\},
\]
and that $T_{S,c}(\pi)=1$ exactly when $c$ is the highest-ranked candidate in
$S$ under ranking $\pi$.  Consequently,
\[
    \sum_{i=1}^n T_{S,c}(\pi_i)
    =\scoreplu(\bpi|_S,c),
\]
so the coordinates of $\sum_iT(\pi_i)$ contain the round scores for every
possible active set $S$.

Each voter sends one DJW-privatized version of the entire vector $T(\pi_i)$.
After summing the reports, the aggregator has estimates
$\widetilde{s}_{S,c}$ for all $(S,c)\in\mathcal I$.  It starts with
$S=\cC$, eliminates a candidate having minimum estimated score
$\widetilde{s}_{S,c}$, and repeats using the coordinates indexed by the new
active set.  Although the realized active set depends on earlier noisy
eliminations, every score that could be queried was included in the original
release.  Thus the voters send no further messages, and the algorithm remains
non-interactive.  The cost is that $T(\pi_i)$ has coordinates for exponentially
many active sets, which produces the $2^{m/2}$ dependence in the utility bound.

\subsection{Upper Bound}

Algorithm \ref{alg:noninteractive-stv} uses the same noisy elimination aggregator as
the interactive protocol.  Its only difference is that it privately releases
the plurality scores for every possible active set in advance.
%Write $a\prec_{\mathrm{tb}}b$ when the fixed STV tie-breaking rule eliminates
%$a$ rather than $b$ in a tie.

\begin{algorithm}[t]
    \caption{Non-Interactive Algorithm for STV}
    \label{alg:noninteractive-stv}
    \begin{algorithmic}[1]
        \REQUIRE Preference profile $\bpi=(\pi_1,\ldots,\pi_n)$, privacy parameter $\varepsilon$ %, failure probability $\gamma$ %, tie-breaking rule $\prec_{\mathrm{tb}}$
        \ENSURE Winner $c\in\cC$
        \STATE Each voter releases $Z_i=\mathcal{R}_{\mathrm{DJW}}(T(\pi_i);r_T,\varepsilon)$
        \STATE Aggregate $\widetilde{s}_{S,c}\gets\sum_{i=1}^n Z_{i,S,c}$ for every $(S,c)\in\mathcal{I}$
        %\STATE $\lambda_T\gets C_{\mathrm{DJW}}\kappa_\varepsilon r_T
        %\sqrt{2n\ln(2D_T/\gamma)}$
        \STATE $\activecandidates \gets\cC$
        \FOR{$j=1$ to $m-1$}
            % \STATE $\widetilde{s}_{\min}\gets\min_{c\in S}\widetilde{s}_{S,c}$
            \STATE $e_j \gets \argmin_{c\in \activecandidates}\widetilde{s}_{\activecandidates ,c}$ \hfill \COMMENT{Tie broken arbitrarily} 
%            \STATE $L\gets\{c\in S: \widetilde{s}_{S,c}\le\widetilde{s}_{\min}+2\lambda_T\}$
%            \STATE $e\gets$ the $\prec_{\mathrm{tb}}$-minimum candidate in $L$
            \STATE $\activecandidates \gets \activecandidates\setminus\{e_j\}$
        \ENDFOR
        \RETURN the unique candidate in $\activecandidates$
    \end{algorithmic}
\end{algorithm}

The following is a more precise version of \Cref{cor:ldp-stv}.

\begin{theorem}[Details for Theorem~\ref{cor:ldp-stv}]
\label{thm:ldp-stv-upper}
There exists a non-interactive $\varepsilon$-LDP algorithm that is
$(\tau,\gamma)$-useful for STV, where
\begin{align*}
    \tau&=m^2 + 2(m-1)\lambda_T+1, \\
    \lambda_T
    &=C_{\mathrm{DJW}}\kappa_\varepsilon r_T
    \sqrt{2n\ln(2D_T/\gamma)}, \\
    r_T&=\sqrt{2^m-m-1},
    \qquad D_T=m2^{m-1}-m.
\end{align*}
In particular, for $0<\varepsilon\le1$,
\[
    \tau=O\left(
    \frac{m2^{m/2}\sqrt{n\ln(m2^m/\gamma)}}{\varepsilon}
    \right).
\]
\end{theorem}

\begin{proof}
We use \Cref{alg:noninteractive-stv}.
Since each voter applies one $\varepsilon$-LDP DJW randomizer, the algorithm
is non-interactive and $\varepsilon$-LDP.  For each coordinate
$(S,c)$, Lemma~\ref{lem:djw-randomizer} shows that
$Z_{i,S,c}-T_{S,c}(\pi_i)$ is sub-Gaussian with parameter at most
$C_{\mathrm{DJW}}\kappa_\varepsilon r_T$.  Independence across voters and a
union bound over the $D_T$ coordinates therefore imply that, with probability
at least $1-\gamma$,
\[
    \max_{(S,c)\in\mathcal{I}}
    |\widetilde{s}_{S,c}-\scoreplu(\bpi|_S,c)|\le\lambda_T.
\]
Condition on this event.

Corollary~\ref{cor:margin-stv-error} shows that an erroneous output requires
\[
    \mov(\bpi,r)\le m^2+2(m-1)\lambda_T<\tau.
\]
Therefore every profile with $\mov(\bpi,r)\ge\tau$ is recovered correctly on
the conditioned event.
\end{proof}

\subsection{STV Lower-Bound Reductions}

The reduction below uses the following certificate to show that the winner of
its constructed STV election is stable under nearby profile changes.

Intuitively, a set $S\in\mathfrak S$ is an active candidate set that may be
reached after a sequence of eliminations that we are willing to allow.  We do
not need to certify one fixed elimination order.  Instead, at every such state
$S$, the set $E(S)$ contains the candidates that may safely be eliminated
next, and deleting any one of them must lead to another certified state in
$\mathfrak S$.  The score-gap condition separates these candidates from the
protected candidates in $S\setminus E(S)$ by more than $2t$.  Since changing
$t$ rankings can alter each round score by at most $t$, this separation
survives every allowed perturbation.  Thus $\mathfrak S$ is a branching family
of possible active sets that starts at $\cC$, always contains $w$, and funnels
every possible elimination sequence to the terminal set $\{w\}$.

\begin{lemma}[STV lower certificate]
\label{lem:margin-stv-lower}
Fix a profile $\bpi$, a candidate $w$, and a nonnegative integer $t$.  Suppose
there is a family $\mathfrak S$ of candidate sets such that
$\cC,\{w\}\in\mathfrak S$ and every $S\in\mathfrak S$ contains $w$.  For
every $S\in\mathfrak S\setminus\{\{w\}\}$, suppose there is a nonempty set of
safe eliminations $E(S)\subseteq S\setminus\{w\}$ satisfying
\[
    S\setminus\{e\}\in\mathfrak S
    \qquad\text{for every }e\in E(S),
\]
and
\[
    \min_{c\in S\setminus E(S)}
    \scoreplu(\bpi|_S,c)
    -
    \min_{e\in E(S)}
    \scoreplu(\bpi|_S,e)
    >2t.
\]
Then $w$ remains the STV winner after any $t$ ranking changes, and hence
\[
    \mov(\bpi,r)>t.
\]
\end{lemma}

\begin{proof}
Changing $t$ rankings changes every restricted plurality score by at most
$t$.  Thus, for every $S\in\mathfrak S$, some candidate in $E(S)$ still has
strictly fewer votes than every candidate in $S\setminus E(S)$.  The candidate
eliminated from $S$ therefore belongs to $E(S)$, regardless of tie-breaking,
and the next active set remains in $\mathfrak S$.  Induction from $\cC$ shows
that the only possible final active set is $\{w\}$.
\end{proof}

We now build a stronger STV gadget for arbitrary $k$-way marginals.

\begin{theorem}[Reduction for Corollary~\ref{cor:main-ldp-stv-kway-edmonds}]
\label{thm:ldp-stv-kway-reduction}
Let $X=(x_1,\ldots,x_N)$ with $x_i\in\{0,1\}^d$, and let
$Q(X)=(Q_A(X))_{A\in\binom{[d]}{k}}$, where
$Q_A(X)=\sum_{i=1}^N\prod_{j\in A}x_{i,j}$.  If there exists an
$\varepsilon/2$-LDP algorithm for STV---with $d+3$ candidates and
$(6k+18)N-2$ voters---that is $(\tau,\gamma)$-useful for $\gamma\le1/4$, then there
exists an $\varepsilon$-LDP algorithm $\cM$ such that
\[
    \Pr\!\left[\|\cM(X)-Q(X)\|_\infty>\tau+1\right]
    \le N\binom{d}{k}\gamma,
\]
and
\[
    \mathbb{E}\!\left[
    \left\|\frac{\cM(X)}{N}-\frac{Q(X)}{N}\right\|_\infty
    \right]
    \le \frac{\tau+1}{N}+N\binom{d}{k}\gamma.
\]
\end{theorem}

\begin{proof}
We construct a compatible profile generator $\cF$.  Let $\nu$ be a padding
parameter, chosen below.
\begin{itemize}\raggedright
    \item $\cF_\mathrm{cand}(d)$ returns
    $\cC=\{c_1,\ldots,c_d,c^*,c',c''\}$.  Let $\sigma$ be a fixed ranking of
    the coordinate candidates $c_1,\ldots,c_d$.
    \item $\cF_\mathrm{data}(X,\varepsilon)$ returns
    $(\bpi_X,\varepsilon')$ with $\varepsilon'=\varepsilon/2$, where
    \[
        \bpi_X=(\pi_1,\ldots,\pi_{2N})\circ(\pi^*,\pi',\pi'')^{2N+\nu}.
    \]
    For each record $x_i$, let
    $S_i=\{c_j:x_{i,j}=1\}$ and $\bar S_i=\{c_j:x_{i,j}=0\}$.  The two
    data-dependent rankings are
    \[
        \pi_{2i-1},\pi_{2i}:
        (\sigma\mid_{\bar S_i})\succ c^*\succ c'\succ c''\succ
        (\sigma\mid_{S_i}).
    \]
    The padding rankings are
    \begin{align*}
        \pi^*:&\quad c^*\succ c'\succ c''\succ\sigma, \\
        \pi':&\quad c'\succ c^*\succ c''\succ\sigma, \\
        \pi'':&\quad c''\succ c^*\succ c'\succ\sigma.
    \end{align*}
    Thus the padding part contains $2N+\nu$ copies of each of
    $\pi^*,\pi'$, and $\pi''$.
    \item $\cF_\mathrm{query}(A)$ returns $(\bpi_A,w_A)$ with $w_A=c^*$ and
    \[
        \bpi_A=\bigcirc_{j\in A}(\pi^{(j)})^{4N+\nu},
    \]
    where
    \[
        \pi^{(j)}:
        c_j\succ c^*\succ c'\succ c''\succ
        (\sigma\mid_{\cC\setminus\{c_j,c^*,c',c''\}}).
    \]
    \item Fix any $j_A\in A$, and let
    \[
        \pi_A^0:
        c_{j_A}\succ c^*\succ c'\succ c''\succ
        (\sigma\mid_{\cC\setminus\{c_{j_A},c^*,c',c''\}}).
    \]
    Then $\cF_\mathrm{thres}$ returns the fixed-size profile
    \[
        \bpi_A^{(\tau)}
        =(\pi',\pi'')^{2\tau-1}\circ(\pi_A^0)^{4(N-\tau)},
    \]
    using the same $\pi'$ and $\pi''$ as above.  The first term contains
    $2\tau-1$ copies of each of $\pi'$ and $\pi''$, so the threshold profile
    contains exactly $4N-2$ rankings for every $\tau\in[N]$.
\end{itemize}

Let $\mu=\sum_{i=1}^N\prod_{j\in A}x_{i,j}$.  Suppose the last remaining
candidates are $\{c^*,c',c''\}\cup\{c_j:j\in A\}$.  At that point the first-place
counts are
% \begin{center}
% \setlength{\tabcolsep}{2pt}
% \resizebox{\columnwidth}{!}{%
% \begin{tabular}{r|cccc}
%     candidate & $c_j$ for $j\in A$ & $c'$ & $c''$ & $c^*$ \\
%     \hline
%     $\#$votes & $\ge 4N+\nu$ & $2N+\nu+2\tau-1$ & $2N+\nu+2\tau-1$ & $2N+\nu+2\mu$
% \end{tabular}
% }
% \end{center}
\begin{center}
\begin{tabular}{c|c}
    candidate & $\#$votes \\
    \hline
    $c_j$ for $j\in A$ & $\ge 4N+\nu$ \\
    $c'$               & $2N+\nu+2\tau-1$ \\
    $c''$              & $2N+\nu+2\tau-1$ \\
    $c^*$              & $2N+\nu+2\mu$
\end{tabular}
\end{center}
Set $\nu=2N$.  Before all coordinate candidates outside $A$ have been
eliminated, each such candidate can receive first-place votes only from the
$2N$ data-dependent rankings.  In contrast, each of $c^*,c',c''$ has at
least $2N+\nu$ first-place votes from the padding, and each queried coordinate
candidate has at least $4N+\nu$ from the query profile.  Hence all coordinate
candidates outside $A$ are eliminated before the displayed set remains.

We first identify the winner.  If $\mu<\tau$, then $c^*$ is eliminated first
from the displayed set and its votes transfer to $c'$.  Candidate $c''$ is
then eliminated and also transfers to $c'$.  At that point $c'$ has
\[
    6N+3\nu+4\tau+2\mu-2
\]
votes, whereas every queried coordinate candidate has at most
$6N+\nu+4(N-\tau)$ votes.  Thus $c'$ wins.  If $\mu\ge\tau$, both $c'$ and
$c''$ are eliminated before $c^*$ and transfer their votes to it.  Candidate
$c^*$ then has at least
\[
    6N+3\nu+2\mu+4\tau-2
\]
votes and again exceeds the same upper bound for every queried coordinate
candidate.  Thus $c^*$ wins if and only if $\mu\ge\tau$.

It remains to certify the margin.  Let $\delta=|\mu-\tau|$ and fix an integer
$t<\delta$; since $\delta\le N=\nu/2$, we also have $t<\nu/2$.  We describe the safe
elimination sets required by Lemma~\ref{lem:margin-stv-lower}.  While an
outside coordinate candidate remains, all such candidates are safe to
eliminate.  Their score is at most $2N$, whereas every protected candidate has
score at least $2N+\nu$, giving a gap greater than $2t$.

If $\mu<\tau$, then at the displayed set only $c^*$ is initially safe: the
next-smallest scores exceed its score by at least $2\delta-1>2t$.  After $c^*$ is
eliminated, every candidate other than the eventual winner $c'$ is safe.
While $c''$ remains, candidate $c'$ leads it by
$2N+\nu+2\mu>2t$.  Once $c''$ is eliminated, $c'$ leads every remaining
queried coordinate candidate by at least
\[
    8\tau+2\mu-2>2t.
\]
These gaps can only increase as further queried candidates are eliminated.

If $\mu\ge\tau$, then $c'$ and $c''$ are initially safe, and every other
candidate leads them by at least $2\delta+1>2t$.  After one is eliminated, its
transfer only increases the lead over the other.  Once both are eliminated,
every queried coordinate candidate is safe, while $c^*$ leads each of them by
at least
\[
    2\mu+8\tau-2>2t.
\]
Again, subsequent transfers go to $c^*$ and preserve the certificate.

The active sets reachable through these safe eliminations form the family
$\mathfrak S$ in Lemma~\ref{lem:margin-stv-lower}.  Therefore the corresponding
winner, $c'$ when $\mu<\tau$ and $c^*$ when $\mu\ge\tau$, is unchanged after
any $t<\delta$ ranking changes.  Hence
$\mov(\bpi,r)\ge|\mu-\tau|$.

The number of voters is
\[
    2N+3(2N+\nu)+k(4N+\nu)+(4N-2)
    =(6k+18)N-2,
\]
so the generator is compatible with an
$\varepsilon/2$-LDP STV algorithm, and Lemma~\ref{lem:k-way-reduction} gives the
claimed marginal estimator.
\end{proof}

Combining this reduction with the $k$-way local marginal lower bound gives the
following consequence.

\begin{proof}[Proof of Corollary~\ref{cor:main-ldp-stv-kway-edmonds}]
Set $\varepsilon_0=1/2$ and apply
Theorem~\ref{thm:ldp-stv-kway-reduction} with its privacy parameter set to
$2\varepsilon$ and $\alpha=\tau$.  The resulting marginal mechanism is
$2\varepsilon$-LDP, so Theorem~\ref{thm:edmonds-local-marginal-lower} applies.

First take $k=2$ and $d=m-3$.  The generated election has
$n=30N-2$, so $N=\Theta(n)$.  The expected-error bound in
Theorem~\ref{thm:ldp-stv-kway-reduction}, together with the two-way marginal
lower bound, gives
\[
    \tau=\widetilde{\Omega}\!\left(
    \min\left\{n,\frac{m\sqrt n}{\varepsilon}\right\}\right).
\]

Next take $k=\lfloor d/2\rfloor=\Theta(m)$.  Now
$n=(6k+18)N-2$, so $N=\Theta(n/m)$, and
\[
    \left(\frac{d}{k}\right)^{\Omega(\sqrt{k})}
    =2^{\Omega(\sqrt m)}.
\]
The assumed bound on $\gamma$, together with
$\binom{d}{k}\le2^d$, makes the failure term
$N\binom{d}{k}\gamma$ negligible.  The $k$-way marginal lower bound therefore gives
\[
    \tau=\widetilde{\Omega}\!\left(
    \min\left\{\frac{n}{m},
    \frac{\sqrt{n/m}}{\varepsilon}2^{\Omega(\sqrt m)}
    \right\}\right).
\]
The factor $m^{-1/2}$ in the second term is absorbed by decreasing the absolute
constant hidden in $\Omega(\sqrt m)$.  Taking the stronger of the two choices
of $k$ proves the claim.
\end{proof}
\section{Additional Details for Interactive Local Differential Privacy}
\label{app:interactive-ldp-details}

This appendix gives the algorithms and complete proofs for the interactive LDP
results in Section~\ref{sec:interactive-ldp}.  The structural margin bounds
used below are collected in Appendix~\ref{app:structural-mov}.  As in the
non-interactive analyses, each proof separates privacy and concentration from
the deterministic margin certificate; interactivity changes only which
statistics are released and their resulting error.

\subsection{Plurality with Runoff}

Algorithm~\ref{alg:interactive-plu-runoff} is the interactive counterpart of
the non-interactive Algorithm~\ref{alg:noninteractive-plu-runoff}.  Both use
half of the privacy budget to estimate the first-round Plurality scores and
select two noisy finalists.  The difference is how they obtain the runoff
statistic.  Because Algorithm~\ref{alg:noninteractive-plu-runoff} must collect
all reports before the finalists are known, it releases noisy margins for
every pair of candidates and then uses the margin for the selected pair.
Algorithm~\ref{alg:interactive-plu-runoff} instead observes the privatized
first-stage result and uses the remaining privacy budget to estimate only the
head-to-head scores of those two finalists.  This adaptive second query avoids
simultaneously estimating all $\binom{m}{2}$ pairwise margins and thereby
improves the dependence of the accuracy bound on the number of candidates.

\begin{algorithm}[H]
    \caption{Interactive Algorithm for Plurality with Runoff}
    \label{alg:interactive-plu-runoff}
    
    \begin{algorithmic}[1]
        \REQUIRE Preference profile $\bpi=(\pi_1,\ldots,\pi_n)$, privacy budget $\varepsilon$
        \ENSURE Winner $c \in \cC$
        \STATE Estimate the scores \((\tscoreplu(\bpi, c))_{c \in \cC}\) using Algorithm~\ref{alg:publish-plurality-scores} with privacy budget \(\eps / 2\)
        \STATE Let \(a\) and \(b\) be the two candidates with the highest scores in \((\tscoreplu(\bpi, c))_{c \in \cC}\)
        \STATE Estimate the scores \((\tscoreplu(\bpi|_{\{a, b\}}, c))_{c \in \{a, b\}}\) using Algorithm~\ref{alg:publish-plurality-scores} with privacy budget \(\eps / 2\)
        \RETURN \(\arg\max_{c \in \{a, b\}} \tscoreplu(\bpi|_{\{a, b\}}, c)\)
    \end{algorithmic}
\end{algorithm}

\begin{proof}[Proof of Theorem~\ref{thm:interactive-runoff-upper}]
    Algorithm~\ref{alg:interactive-plu-runoff} makes two calls to
    Algorithm~\ref{alg:publish-plurality-scores}, each with privacy budget
    $\eps/2$, so sequential composition gives total privacy budget $\eps$.

    Let
    $\lambda=C_{\mathrm{DJW}}\kappa_{\eps/2}
    \sqrt{2n\log(4m/\gamma)}$.  Apply
    Corollary~\ref{cor:acc-algo-plu} to the first-stage release with failure
    probability $\gamma/2$.  Conditional on its transcript, the selected pair
    $\{a,b\}$ is fixed, and the second-stage release uses fresh randomness.
    The same corollary, now applied conditionally with failure probability
    $\gamma/2$, therefore controls the two second-stage scores.  Since $m\ge2$,
    the displayed value of $\lambda$ is sufficient for both stages.  Thus, by
    a union bound, with probability at least $1-\gamma$ every estimate used by
    the algorithm is within $\lambda$ of its true value.  Condition on this
    event.

    The two second-stage score errors imply that the estimated pairwise margin
    has error at most $2\lambda$.  Lemma~\ref{lem:margin-runoff}, with
    $\eta=\lambda$ and $\zeta=2\lambda$, shows that an erroneous output requires
    $\mov(\bpi,r)\le2\lambda+2$.
\end{proof}

\subsection{Single Transferable Vote}

The interactive algorithm for STV is given in Algorithm~\ref{alg:interactive-stv}.
In an STV round with active candidate set $S$, we call
$\scoreplu(\bpi|_S,c)$ the round score of candidate $c\in S$.

\begin{algorithm}[H]
    \caption{Interactive Algorithm for STV}
    \label{alg:interactive-stv}
    
    \begin{algorithmic}[1]
        \REQUIRE Preference profile $\bpi=(\pi_1,\ldots,\pi_n)$, privacy budget $\varepsilon$
        \ENSURE Winner $c \in \cC$
        \STATE $\activecandidates \gets \cC$
        \FOR{$j=1$ to $m-1$}
            \STATE Estimate the scores \((\tscoreplu(\bpi|_{\activecandidates}, c))_{c \in \activecandidates}\) using Algorithm~\ref{alg:publish-plurality-scores} with privacy budget \(\eps / (m-1)\)
            \STATE $e_j \gets \arg\min_{c \in \activecandidates}
            \tscoreplu(\bpi|_{\activecandidates}, c)$
            \STATE $\activecandidates \gets \activecandidates \setminus \{e_j\}$
        \ENDFOR
        \RETURN the unique candidate in $\activecandidates$
    \end{algorithmic}
\end{algorithm}

The aggregation rule is the same noisy elimination rule analyzed in
Corollary~\ref{cor:margin-stv-error}; only the way the round scores are
obtained differs from the non-interactive algorithm.

\begin{proof}[Proof of Theorem~\ref{thm:interactive-stv-upper}]
    Algorithm~\ref{alg:interactive-stv} makes $m-1$ calls to
    Algorithm~\ref{alg:publish-plurality-scores}, each with privacy budget
    $\eps/(m-1)$, so composition gives total privacy budget $\eps$.

    Let
    $\lambda=C_{\mathrm{DJW}}\kappa_{\eps/(m-1)}
    \sqrt{2n\log(2m^2/\gamma)}$.  In each round, condition on the transcript
    before that round.  The active candidate set is then fixed, and the round
    uses fresh randomizer randomness, so
    Corollary~\ref{cor:acc-algo-plu} applies conditionally with failure
    probability at most $\gamma/(m-1)$.  A union bound over the rounds therefore
    shows that, with probability at least $1-\gamma$, for every
    realized active set $S$ and every $c\in S$,
    \[
        |\tscoreplu(\bpi|_S,c)-\scoreplu(\bpi|_S,c)|\le\lambda.
    \]
    Condition on this event.

    Corollary~\ref{cor:margin-stv-error} shows that an erroneous output requires
    \[
        \mov(\bpi,r)
        \le m^2+2(m-1)\lambda.
    \]
\end{proof}

\subsection{Condorcet Voting}

For a collection $\mathcal{P}$ of $p$ two-candidate contests, let
$B_{\mathcal{P}}(\pi)$ be the concatenation of the $p$ one-hot vectors that
record the top-ranked candidate in each restricted ranking.  Since
$\|B_{\mathcal{P}}(\pi)\|_2=\sqrt p$, one release
\[
    \mathcal{R}_{\mathrm{DJW}}
    (B_{\mathcal{P}}(\pi);\sqrt p,\rho)
\]
simultaneously estimates all $2p$ plurality scores under $\rho$-LDP.

\begin{algorithm}[h]
\caption{Interactive Algorithm for Condorcet Winner}
\label{alg:interactive-condorcet}

\begin{algorithmic}[1]
    \REQUIRE Preference profile $\bpi = (\pi_1,\ldots,\pi_n)$, privacy budget $\varepsilon$
    \ENSURE Winner $c \in \cC \cup \{\bot\}$
    \STATE $C' \gets \cC, Q \gets \lceil\log_2 m\rceil + 1$
    \WHILE{$|C'| > 1$}
        \STATE Partition $C'$ into a collection $\mathcal{P}$ of disjoint pairs and, if $|C'|$ is odd, one unpaired candidate $u$
        \STATE Estimate all scores
        \(
            \left( \tscoreplu(\bpi|_{\{a,b\}},d) \right)_{\substack{\{a,b\}\in\mathcal{P}\\ d\in\{a,b\}}}
        \)
        using one batched DJW release
        with privacy budget $\varepsilon/Q$
        \FOR{each $\{a,b\}\in\mathcal{P}$}
            \STATE $c \gets \arg\min_{d\in\{a,b\}} \tscoreplu (\bpi|_{\{a,b\}},d)$
            \STATE $C' \gets C' \setminus \{c\}$
        \ENDFOR
    \ENDWHILE
    
    \STATE Let $x$ be the unique candidate in $C'$
    \STATE Estimate all scores
        \(
            \left( \tscoreplu(\bpi|_{\{x,c\}},d) \right)_{\substack{c\in \cC\setminus\{x\}\\ d\in\{x,c\}}}
        \)
        using one batched DJW release
        with privacy budget $\varepsilon/Q$
    \FOR{each $c \in \cC \setminus \{x\}$}
        \IF{$\tscoreplu(\bpi|_{\{x,c\}},x)
        \le
        \tscoreplu(\bpi|_{\{x,c\}},c)$}
            \RETURN $\bot$
        \ENDIF
    \ENDFOR
    \RETURN $x$
\end{algorithmic}
\end{algorithm}

We combine the structural bound from Lemma~\ref{lem:margin-condorcet} with the
accuracy of the batched pairwise releases.  An erroneous noisy comparison
yields a small true pairwise margin, which the lemma converts into an upper
bound on the election's margin of victory.

\begin{proof}[Proof of Theorem~\ref{thm:interactive-condorcet-upper}]
    Algorithm~\ref{alg:interactive-condorcet} makes $Q$ batched DJW releases,
    each with privacy budget $\eps/Q$, so composition gives total privacy budget
    $\eps$.

    Every batch contains at most $m-1$ contests, so its input vector has $\ell_2$-norm at most $\sqrt m$.  Condition on the transcript before any batch.  Its
    contests are then fixed, and its fresh reports satisfy
    Lemma~\ref{lem:djw-randomizer} conditionally.  There are at most $Q\le m$
    batches and at most $2m$ released coordinates per batch.  Applying the
    conditional tail bound and a union bound over these at most $2m^2$
    coordinates shows that, with probability at least $1-\gamma$, every
    plurality-score estimate used by the algorithm has error at most
    \[
        \lambda=C_{\mathrm{DJW}}\kappa_{\eps/Q}
        \sqrt{2mn\log(4m^2/\gamma)}.
    \]
    Consequently, with
    \[
        \hat{D}(a,b)
        =\tscoreplu(\bpi|_{\{a,b\}},a)
        -\tscoreplu(\bpi|_{\{a,b\}},b),
    \]
    we have
    \[
        |\hat{D}(a,b) - D_{\bpi}(a,b)| \leq 2\lambda.
    \]
    
    We argue by contraposition.  If the algorithm changes the outcome, then
    either a true Condorcet winner $w$ loses one noisy head-to-head comparison, or
    a profile with no Condorcet winner produces a candidate $c$ that wins all
    noisy head-to-head comparisons.  In the first case,
    $\min_{d\in\cC\setminus\{w\}}D_{\bpi}(w,d)\le2\lambda$.  In the second,
    $\max_{a\in\cC\setminus\{c\}}D_{\bpi}(a,c)\le2\lambda$.

    Lemma~\ref{lem:margin-condorcet} then bounds the margin of victory by
    \[
        1+3\log_2 m
        +6C_{\mathrm{DJW}}\kappa_{\eps/Q}
        \sqrt{2mn\log\!\left(\frac{4m^2}{\gamma}\right)}.
    \]
\end{proof}

\subsection{Interactive lower bound}

\begin{proof}[Proof of Theorem~\ref{thm:interactive-plurality-lower}]
Suppose for contradiction that an interactive protocol is
$\bigl(C\min\{n,\sqrt n/\varepsilon\},\gamma\bigr)$-useful, where the universal
constant $C>0$ will be fixed below.
The fully interactive simple-hypothesis-testing lower bound
of~\citet[Theorem~5.3]{JosephMNR19} implies that there is a universal
$c_0>0$ such that distinguishing two distributions at total-variation
distance $\alpha$ with error at most $1/3$ requires at least
$c_0/(\varepsilon^2\alpha^2)$ users, for $0<\varepsilon\le1$.

First suppose $n\varepsilon^2<c_0/2$.  Take $P_+$ and $P_-$ to be the point
masses on the rankings $c_1\succ c_2$ and $c_2\succ c_1$, respectively.
The resulting profiles are unanimous and have margin at least $n/2$.  A
$\bigl(C\min\{n,\sqrt n/\varepsilon\},\gamma\bigr)$-useful protocol, for
$C\le1/2$ and $\gamma\le1/6$, would
distinguish $P_+^{\otimes n}$ from $P_-^{\otimes n}$ with error at most $1/6$,
contradicting the testing lower bound with $\alpha=1$.

Now suppose $n\varepsilon^2\ge c_0/2$.  Let
\[
    \alpha=\frac{a}{\varepsilon\sqrt n},
\]
where $a>0$ is a sufficiently small universal constant, and let $P_+$ and
$P_-$ be the distributions on the two rankings defined by
\begin{align*}
    P_+(c_1\succ c_2)&=P_-(c_2\succ c_1)=\frac{1+\alpha}{2}, \\
    P_+(c_2\succ c_1)&=P_-(c_1\succ c_2)=\frac{1-\alpha}{2}.
\end{align*}
Their encoded means are $\alpha$ and $-\alpha$.  Choosing
$a\le\sqrt{c_0/8}$ ensures $\alpha\le1/2$ and also makes
$n<c_0/(\varepsilon^2\alpha^2)$, so no fully interactive
$\varepsilon$-LDP test distinguishes the two hypotheses with error at most
$1/3$.

Choose the universal constant $\varepsilon_0>0$ sufficiently small that
\[
    \exp\!\left(-\frac{a^2}{8\varepsilon_0^2}\right)\le\frac1{12}.
\]
For $0<\varepsilon\le\varepsilon_0$, Hoeffding's inequality shows that, under
either hypothesis, with probability at least $11/12$ the Plurality winner
agrees with the sign of the mean and the score gap is at least
$n\alpha/2$.  The margin of victory is then at least
$n\alpha/4=a\sqrt n/(4\varepsilon)$.  Since
$\min\{n,\sqrt n/\varepsilon\}\le\sqrt n/\varepsilon$, taking $C\le a/4$
makes this margin at least the assumed usefulness threshold.  The protocol
would therefore yield a test with error at most $1/12+1/6=1/4$, again a
contradiction.

Taking $C\le\min\{1/2,a/4\}$ proves the two regimes simultaneously.
\end{proof}

\end{document}